%% file: long_manuscript.tex
\RequirePackage{lineno}
\newif\ifpreprint%
\preprinttrue%
\ifpreprint%
	\documentclass[twocolumn,letterpaper,english,aps,pre,floatfix,superscriptaddress]{revtex4}
\else%
	\documentclass[twocolumn,letterpaper,english,aps,pre,floatfix,amssymb,linenumbers,superscriptaddress]{revtex4}
\fi%

\usepackage{xr}
\usepackage{amsmath}
\usepackage{dsfont}
\usepackage{xcolor}
\usepackage{amssymb}
\usepackage{graphicx}
\usepackage{algorithm}
\usepackage{algpseudocode}
\usepackage{braket}
\usepackage{units}
\usepackage{siunitx} 
\usepackage{bm}
\usepackage{amsthm}
\usepackage{float}
\theoremstyle{plain}
\theoremstyle{definition}

\newtheorem{lemma}{Lemma}[]

\DeclareMathOperator*{\argmax}{argmax} 

\ifx\pdftexversion\undefined%
\usepackage[dvips]{hyperref}
\else%
\usepackage{hyperref}
\fi%
\hypersetup{
	colorlinks = false,
	hypertexnames=false
}
\newcommand{\ssm}{\scriptscriptstyle\rm}

\renewcommand{\theta}{\vartheta}
\renewcommand{\phi}{\varphi}

\begin{document}
	
	\title{Symmetries and phase diagrams with real-space mutual information neural estimation}

	\date{\today}
	
	\author{Doruk Efe G\"okmen}
	\affiliation{Institute for Theoretical Physics, ETH Zurich, 8093 Zurich, Switzerland}
	\author{Zohar Ringel}
	\affiliation{Racah Institute of Physics, The Hebrew University of Jerusalem, Jerusalem 9190401, Israel}
	\author{Sebastian D. Huber}
	\affiliation{Institute for Theoretical Physics, ETH Zurich, 8093 Zurich, Switzerland}
	\author{Maciej Koch-Janusz}
	\affiliation{Institute for Theoretical Physics, ETH Zurich, 8093 Zurich, Switzerland}
	\affiliation{Department of Physics, University of Zurich, 8057 Zurich, Switzerland}
	\affiliation{James Franck Institute, The University of Chicago, Chicago, Illinois 60637, USA}


	\begin{abstract}	
	Real-space mutual information (RSMI) was shown to be an important quantity, formally and from a numerical standpoint, in finding coarse-grained descriptions of physical systems. It very generally quantifies spatial correlations, and can give rise to \emph{constructive} algorithms extracting relevant degrees of freedom. Efficient and reliable estimation or maximization of RSMI is, however, numerically challenging. A recent breakthrough in theoretical machine learning has been the introduction of variational lower bounds for mutual information, parametrized by neural networks. Here we describe in detail how these results can be combined with differentiable coarse-graining operations to develop a single unsupervised neural-network based algorithm, the RSMI-NE, efficiently extracting the relevant degrees of freedom in the form of the operators of effective field theories, directly from real-space configurations. 
	We study the information contained in the \emph{statistical ensemble} of constructed coarse-graining transformations, and its recovery from partial input data using a secondary machine learning analysis applied to this ensemble.	In particular, we show how symmetries, also emergent, can be identified. We demonstrate the extraction of the phase diagram and the order parameters for equilibrium systems, and consider also an example of a non-equilibrium problem.
	\end{abstract}
	
	\maketitle
	
	\section{Introduction}
	Constructing coarse-grained descriptions of physical systems is a fundamental operation, both practically and from the foundational theory viewpoint. It is often difficult to simulate a complex systems from first principles, even if a microscopic model exists, necessitating layers of descriptions and independent simulations of properties at differing scales.\cite{pt_multiscale} This is the practical counterpart to the powerful methodology of deriving effective theories providing a summary of physics at this scale, which may involve qualitatively new emergent properties.
	
	The above idea finds its full realization within the framework of the renormalization group (RG),\cite{PhysicsPhysiqueFizika.2.263,Wilson1974,Wilson1975,Fisher1998} where the coarse-graining transformation of local degrees of freedom (DOFs) gives rise to the RG-flow in the space of theories. 
	While momentum-space methods had a profound impact on physics,\cite{PhysRevB.4.3184,PhysRevLett.30.1343} real-space renormalization, despite some veritable successes,\cite{PhysRevLett.69.534,PhysRevB.51.6411,PhysRevB.50.3799,IGLOI2005277} has not yet reached a similar status. 
	In practice, real-space procedures are non-trivial to design, involve ad-hoc choices, and can rarely be executed exactly, amplifying any approximation as they are iterated. Moreover, analytical understanding is often lacking. At the same time disordered and complex systems, where the notion of momentum may not be available, are naturally amenable to real-space approaches and thus provide a strong motivation for development of improved methods.

	The arbitrariness of the real-space RG transformation choice can be drastically reduced: an optimal coarse graining rule \emph{for a given system} exists,\cite{Koch-Janusz2018,optimalRSMI} theoretically defined by maximising the information the coarse DOFs retain about distant parts of the system \emph{i.e.}~the real-space mutual information (RSMI). Such coarse-grained representation is an optimal (lossy) compression of information about long-distance properties of the system. 
	This observation can, however, be made much more powerful. Since long-range information is due to the scaling operators, its optimal compression not only defines a better iterative RG rule, but, performed \emph{for each point in the phase diagram}, should allow to \emph{directly} extract all the operators themselves, without explicitly executing the RG flow. This was formally proven, at least for critical systems in equilibrium.\cite{Gordon2020} Conceptually, this opens the possibility of identifying formal components of the effective theory directly from real-space data, using only its statistical properties. 
	
		\begin{figure}
		\centering
		\includegraphics[width=\linewidth]{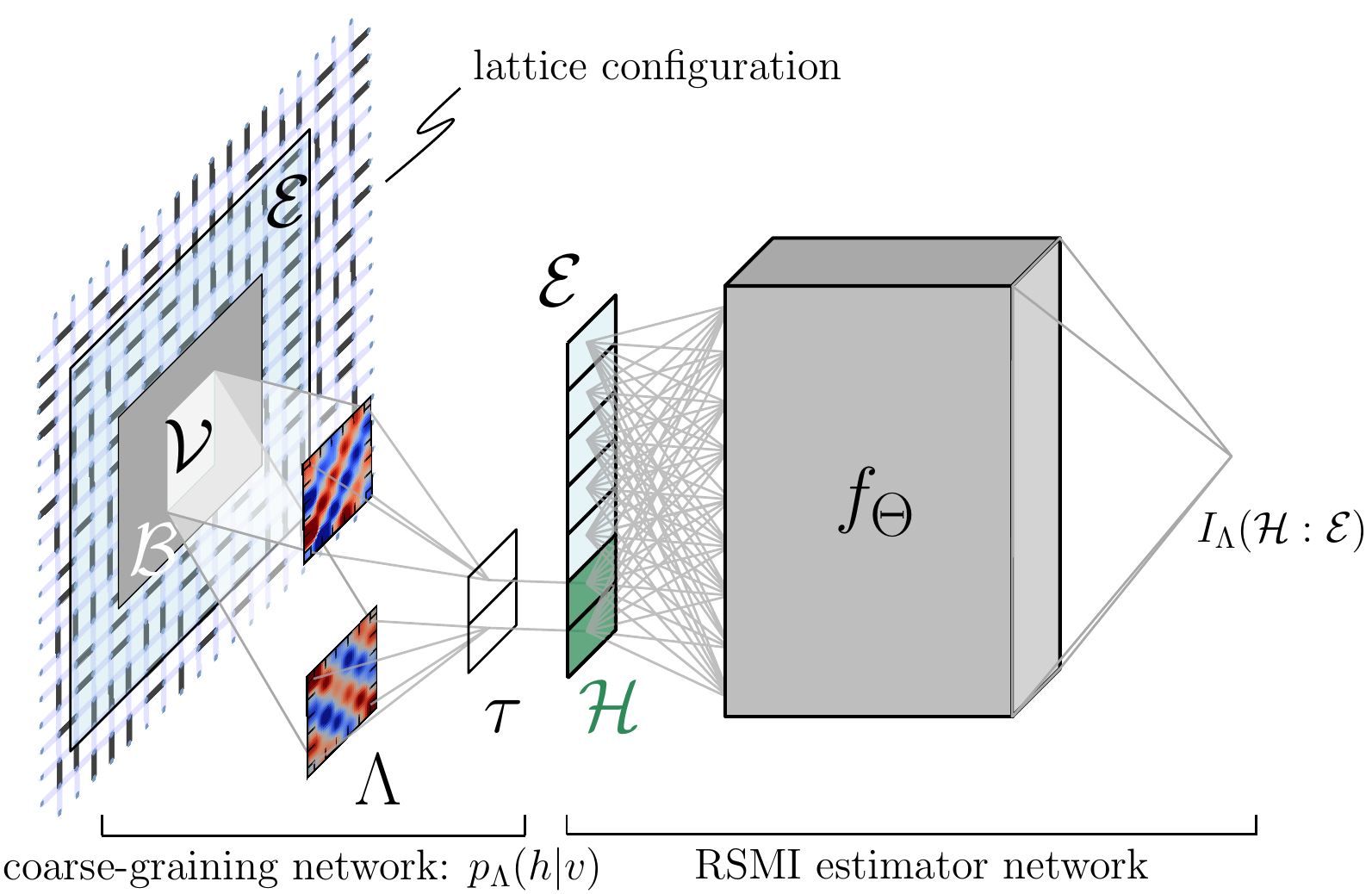}
		\caption{{\bf The architecture of RSMI-NE.} The coarse-graining transformation extracts the order parameters and relevant operators via the transformations $\Lambda$ and discretising step $\tau$. The long-range information $I_{\Lambda}(\mathcal{H}:\mathcal{E})$ which $\Lambda$ maximize is estimated by $f_\Theta$, all of which are parametrized by neural networks and co-trained together.}
		\label{fig:netfig}
	\end{figure}
	
	The core quantity of interest, the RSMI, is however challenging to maximize or estimate, as is mutual information (MI) in general,\cite{poole2019variational} which in practice limits the applicability of such general approach to but the simplest of systems.

	Here we overcome this limitation by developing a highly efficient algorithm, the RSMI neural estimator (RSMI-NE), computing the optimal coarse-grainings. Our approach is based on state-of-art machine learning techniques of estimating mutual information by maximising its rigorous lower-bounds.\cite{belghazi2018mine,poole2019variational} A key algorithmic idea, introduced in Ref.~\onlinecite{belghazi2018mine}, is to parametrize these bounds by a neural network $f_\Theta$, optimizing over its parameters $\Theta$. We use this differentiable variational \emph{ansatz} for RSMI to optimize the RG transformation, expressed by another neural network \emph{ansatz} with parameters $\Lambda$. 
	Crucially, both networks are combined and trained together using stochastic gradient descent back-propagating through the whole structure (see Fig.~\ref{fig:netfig}), even in the presence of discrete-valued coarse-graining.
	The algorithm is several orders of magnitude faster than RSMI estimation in Ref.\onlinecite{Koch-Janusz2018}, allowing to explore large systems and length-scales, demonstrating superior convergence and stability. 
	
	The methods presented here were advertised in a companion work,\cite{rsmine_letter} where the emphasis was on the 
	extraction of relevant operators on the lattice. Here, in contrast, we have two other areas of focus. First, we present and examine in detail the algorithmic aspects of the method, and the ML/statistics tools it uses.
	Second, we study the properties of a novel object: the \emph{ensemble} of the coarse-graining transformations. Taking a step beyond identifying relevant operators in a single instance of RSMI optimisation, we demonstrate that their distribution, generated by multiple independent optimisation runs, reveals information about the symmetries of the system, broken and  \emph{emergent}. We consider application of the RSMI formalism to further systems, including non-equilibrium ones. 
	
	The manuscript is organized as follows: In Sec.~\ref{RSMINE} we describe the main theoretical and algorithmic components of the RSMI-NE. Specifically, in Sec.~\ref{RSMINE}.A the general idea of the RSMI approach is briefly described, Sec.~\ref{RSMINE}.B reviews the neural-network based variational lower bounds on MI which are used and implemented in the ${\textsf {RSMI-NE}}$ package\footnote{Source code for the RSMI-NE is available online at \url{https://github.com/RSMI-NE/RSMI-NE}.} we make available. Section~\ref{RSMINE}.C describes the parametrisation of the coarse-graining as a neural network, and ensuring its \emph{differentiability} for \emph{discrete} latent variables, Sec.~\ref{RSMINE}.D combines these elements into the complete algorithm and discusses convergence properties. Turning to physics, in
	Sec.\ref{equilibrium} we detail the physical information (phase diagram, operators, symmetries) contained in the algorithm outputs, testing it in equilibrium statistical systems. Particularly, in Sec.\ref{equilibrium}C, we show that the ensemble of coarse-graining transformations allows to identify the symmetries of the system, also \emph{emergent} ones, and retrieve them from incomplete data using a secondary ML analysis of the ensemble.	
	The possible extension to non-equilibrium problems is discussed in Sec.\ref{section:chipping}, and the algorithm is validated on the example of the non-equilibrium chipping model. 
	We conclude in Sec. \ref{conclusion} with a brief discussion of the scope of the method, and its possible extensions and  applications.
	The Supplemental Material gives technical details related to the code and data generation.\footnote{See Supplemental Material for more details of the RSMI-NE algorithm and the directed loop Monte Carlo algorithm for the dimer model, which includes Refs.~\onlinecite{10.1093/imamat/6.3.222, BFGS, PhysRevB.73.144504, kenyon1999trees}}
	
	\section{The RSMI-NE algorithm}\label{RSMINE}
	
	\subsection{The RSMI variational principle}
	RG is rooted in the observation that most microscopic details are irrelevant for large-scale behaviour of physical systems. It is, however, necessary to define a firm basis for determining exactly which short-scale details are projected out in a \emph{coarse-grained} description. This has proven difficult in real-space.\cite{vanEnter1993}
	To address this issue, Ref.~\onlinecite{Koch-Janusz2018} proposed that the optimal real-space RG transformation maximises an information theoretical quantity, the real-space mutual information (RSMI), measuring the information shared between a coarse-grained degree of freedom and its distant environment at the original fine level. 
	This constitutes a universal principle for determining the coarse-grained description for any statistical system. Moreover, it gives \emph{direct} access to the relevant operators on the lattice.\cite{Gordon2020,rsmine_letter}
	
	Consider a system of classical DOFs in any dimension denoted by a multi-variate random variable $\mathcal{X}$, whose physics is encoded in a probability measure $p(x)$, either Gibbsian, \emph{i.e.} $p(x)\propto e^{-\beta H(x)}$, or a generic non-equilibrium distribution. Here we denote by $x\sim p(x)$ an instance of the random variable $\mathcal{X}$ drawn from the distribution $p(x)$. A coarse-graining rule $\mathcal{X}\to \mathcal{X}'$ is defined as a conditional distribution $p_\Lambda(x'|x)$, determined by a set of parameters $\Lambda$ to be optimised. It is a probabilistic map generating a particular compressed representation of the original DOFs. 
	
	A coarse-graining is typically carried out on disjoint spatial blocks $\mathcal{V}_i \subset \mathcal{X}$, and it factorises:  $p(x'|x)=\prod_i p_{\Lambda_i}(h_i|v_i)$, such that $\mathcal{X}=\bigcup_i \mathcal{V}_i$ and $\mathcal{X}'=\bigcup_i \mathcal{H}_i$, with $p_{\Lambda_i}(h_i|v_i)$ the coarse-graining rule applied to block $i$. If the system is translation invariant a fixed $\Lambda_i \equiv \Lambda$ suffices; otherwise, \emph{e.g.}~in disordered systems, it can be favourable to optimise each block individually.
	
	The RSMI approach identifies coarse-graining filters extracting the most relevant long-range features as the ones retaining the most information shared by a block $\mathcal{V} \subset \mathcal{X}$ to be coarse-grained, and its distant environment $\mathcal{E}$,\cite{Koch-Janusz2018,optimalRSMI} \emph{i.e.}~those that optimally \emph{compress} this information. The environment is separated from $\mathcal{V}$ by a shell of non-zero thickness constituting the buffer $\mathcal{B}$, and forms the remainder of the system (see Fig.~\ref{fig:netfig}.a). The ``shared information" between the random variables $\mathcal{H}$ and $\mathcal{E}$ is formalised by the Shannon mutual information:
	\begin{equation}\label{eq:RSMI_def}
		I_{\Lambda}(\mathcal{H}:\mathcal{E}) = \sum_{h,e}p_\Lambda(e,h) \log\left(\frac{p_\Lambda(e,h)}{p_\Lambda(h)p(e)}\right),
	\end{equation}
	where $p_\Lambda(e,h)$ and $p(h)$ are the marginal probability distributions of $p_\Lambda(h,x)=p_\Lambda(h|v)p(x)$ obtained by summing over the DOFs in $\{\mathcal V$, $\mathcal B\}$ and $\{\mathcal V, \mathcal B, \mathcal E\}$, respectively. The size of the buffer $\mathcal{B}$ sets the RG scale, and acts a filter, only allowing the information about large-scale properties to be compressed into the coarse-grained variables.
	Finding the optimal coarse-graining is thus formulated as a variational principle maximizing $I_\Lambda$ as a function of parameters $\Lambda$.
	
	Maximizing mutual information for high-dimensional random variables is known to be difficult,\cite{PhysRevE.69.066138, PhysRevE.52.6841} which limited the usefulness of the RSMI approach.\cite{Koch-Janusz2018} This challenge can now be overcome with the help of recent ML results combining mathematically rigorous variational bounds on mutual information \cite{doi:10.1002/cpa.3160360204,NIPS2003_2410,5605355} with deep learning.\cite{belghazi2018mine,poole2019variational} In the following section we describe in detail the two components forming the core of the new fast RSMI-NE algorithm: efficient neural MI estimation, and differentiably parametrising the coarse-graining operation. 
	
	\subsection{Differentiable lower-bounds of RSMI}\label{mi_bounds}
	We follow the recent approach of Refs.~\onlinecite{belghazi2018mine, poole2019variational}. Given possibly high-dimensional random variables $\mathcal{X}$, $\mathcal{Y}$, a variational upper or lower bound for $I(\mathcal{X}:\mathcal{Y})$ is constructed, and parametrised by a sufficiently expressive non-linear neural network (NN) \emph{ansatz} $f(x,y)$ modelling the statistical dependence of $\mathcal{X}$ and $\mathcal{Y}$. The weights of the network are updated in an unsupervised learning scheme using the joint samples of $\mathcal{X}$, $\mathcal{Y}$ (in our case: $\mathcal{H}$, $\mathcal{E}$), producing a sequence of differentiable bounds $I_\Lambda$, asymptotically exact.

	It is possible to either minimise a variational upper-bound or to maximise a lower-bound of MI. Given our central aim of maximising RSMI with respect to the parameters $\Lambda$ of some coarse-graining network, we focus on the latter. We mainly use the noise-contrastive lower-bound of MI (InfoNCE), a multi-sample bound characterised by lower variance, but we also review the single-sample bounds it is the extension of. As we shall see, the general form of these bounds is motivated by the interpretation of MI as distinguishing between independently and jointly distributed random variables.
	
	\subsubsection{Single-sample lower-bounds}
	With this motivation in mind, and MI defined as follows:
	\begin{equation*}
		I(\mathcal{X} : \mathcal{Y}) = \mathbb{E}_{p(x,y)}\left[\log \frac{p(x|y)}{p(x)}\right] = \mathbb{E}_{p(x,y)}\left[\log \frac{p(y|x)}{p(y)}\right],
	\end{equation*}
	we introduce the conditional probability distribution $q(x|y)$ as a variational \textit{ansatz} approximating $p(x|y)$. We shall first keep the form of $q(x|y)$ unconstrained, and derive a lower-bound of $I(\mathcal{X} : \mathcal{Y})$ in its terms. Our goal is to find the optimal \textit{ansatz} that makes the bound tight. Subsequently, we explain how the form of $q(x|y)$ can be constrained at the onset to improve the corresponding lower-bound, yielding a more tractable estimator.
	
	Since the Kullback-Leibler (KL) divergence between them:
	\begin{align}
		D_{\rm{KL}}(p(x|y)||q(x|y))& = \mathbb{E}_{p(x|y)}\left[\log \frac{p(x|y)}{q(x|y)}\right]
	\end{align}
	is non-negative, we immediately obtain a lower-bound for $I(\mathcal{X}:\mathcal{Y})$, known as the Barber-Agakov (BA) bound:\cite{NIPS2003_2410}
	\begin{align}
		I(\mathcal{X}:\mathcal{Y})\geq& \mathbb{E}_{p(x,y)}\left[\log \frac{q(x|y)}{p(x)}\right]\\
		&=\mathbb{E}_{p(x,y)}\left[\log q(x|y)\right]+H(\mathcal{X})=:I_{\rm BA}(\mathcal{X}:\mathcal{Y}),\nonumber
	\end{align}
	where $H(\mathcal{X})$ is the entropy of $\mathcal{X}$. This bound is a functional of the \emph{ansatz}: $I_{\rm BA}(\mathcal{X}:\mathcal{Y})=I_{\rm BA}(\mathcal{X}:\mathcal{Y})[q(x|y)]$. Since $D_{\rm KL}=0$ if and only if $q(x|y)=p(x|y)$, the BA bound is tight only when the ansatz $q(x|y)$ equals $p(x|y)$. 
	
	The observation that mutual information measures the correlations between variables motivates the idea that in modelling $p(x|y)$ the \emph{ansatz} $q(x|y)$ should focus on the dependencies between the variables $\mathcal{X}$ and $\mathcal{Y}$. Consider thus the following \emph{ansatz} family:
	\begin{equation}\label{EBA_ap}
		q(x|y):=\frac{p(x)}{Z(y)}e^{f(x,y)},
	\end{equation}
	with $Z(y):=\mathbb{E}_{p(x)}\left[e^{f(x,y)}\right]$. In the above energy-based form, the complex correlations within the possibly high-dimensional data $\mathcal{X}$ are contained in the marginal distribution $p(x)$. The resulting lower-bounds are sensitive mainly to the variables' interdependency. In other words, maximising the lower-bound of MI is rephrased as a search for a  ``critic"\cite{poole2019variational} function $f(x,y)$ modelling the relationships, between $\mathcal{X}$ and $\mathcal{Y}$ very well. The critic function, distinguishing the ``positive" samples from the joint distribution, from the ``negative" ones generated by the product of marginals, will be approximated by a neural network.

	Substituting the energy-based \textit{ansatz} into the BA bound, we obtain the \textit{unnormalised} BA bound (UBA):
	\begin{equation}\label{UBA}
		I_{\rm UBA}(\mathcal{X}:\mathcal{Y}):=\mathbb{E}_{p(x,y)}[f(x,y)] - \mathbb{E}_{p(y)} [\log Z(y)].
	\end{equation}
	By the same arguments as above, the UBA bound is tight when $\frac{p(x)}{Z(y)}\exp{f(x,y)}=p(x|y)$. \footnote{This condition is equivalent to $f(x,y)=\log p(y|x)+c(y)$, because $\exp{f(x,y)}=p(x|y)\frac{Z(y)}{p(x)}=p(y|x)\frac{ Z(y)}{p(y)}$
	with $c(y):=\log Z(y) -\log p(y)$ being a constant in $x$.} Taking advantage of the strict concavity of the $\log$ function, one arrives at the \emph{tractable} version of UBA bound:\cite{poole2019variational}
	\begin{align}
		 I_{\rm TUBA}(\mathcal{X}:\mathcal{Y}):=&\mathbb{E}_{p(x,y)}[f(x,y)] \nonumber\\
		 &- \mathbb{E}_{p(x)p(y)}\left[\frac{e^{f(x,y)}}{a(y)}\right] - \mathbb{E}_{p(y)}\left[\log \frac{a(y)}{e}\right]. \nonumber
	\end{align}
	In several studies, see \emph{e.g.}~Nguyen, Wainwright and Jordan (NWJ) \cite{5605355}, $f$-GAN by Nowozin \textit{et al.}~\cite{nowozin2016fgan} and MINE-$f$ by Belghazi \textit{et al.},\cite{belghazi2018mine} the \textit{baseline} function $a(y)$ is fixed to be the constant $e$. This choice simplifies the TUBA bound:
	\begin{equation}\label{NWJ_bound}
		I_{\rm NWJ}(\mathcal{X}:\mathcal{Y}):=\mathbb{E}_{p(x,y)}[f(x,y)] - e^{-1}\mathbb{E}_{p(x)p(y)}[e^{f(x,y)}] .
	\end{equation}
	Note that in this case $f(x,y)$ should be optimised under the constraint that $q(x|y)$ is normalised. The MINE approach\cite{belghazi2018mine} has also recently been used to estimate entropy in physical systems.~\cite{Nir30234}

	\subsubsection{Replica lower-bounds}
	Despite the improvement due to the energy based \emph{ansatz}, the above ``single-sample" bounds are known to suffer from a large variance.\cite{poole2019variational, oord2018representation} An improved approach is to divide a single batch of samples (\emph{e.g.}~Monte Carlo) for the pair of random variables $(\mathcal{X},\mathcal{Y})$ into minibatches of $K$-fold ``replicated" random variables $(\mathcal{X}_i,\mathcal{Y}_i)_{i=1}^K$, and to derive the corresponding ``multi-sample" lower-bounds (the confusing dual usage of the term ``sample" is standard). These are obtained by taking the average of the single-sample bounds, and address the issue of large variance by means of noise-contrastive estimation (NCE)\cite{899a65b4919f47c8a06d115df85dad11}, first proposed in the context of MI estimation in Ref.~\onlinecite{oord2018representation}. 
	
	A multi-sample bound estimates $I(\mathcal{X}_1,\mathcal{Y})$, where $(\mathcal{X}_1,\mathcal{Y})\sim p(x_1,y)$, 
	given $K-1$ additional independent ``replicas" for one of the random variables, say $\mathcal{X}$ (drawn from the marginal distribution). We denote them by $\mathcal{X}_{2:K}\sim \prod_{j=2}^K p(x_j)$.
	All of the $K$ independent replicas of the random variable $\mathcal{X}$ are treated as a single $K$-dimensional random variable $\mathcal{X}_{1:K}$, and it is easily seen that $I(\mathcal{X}_1:\mathcal{Y}) = I(X_{1:K}:Y)$ since only $\mathcal{X}_1$ was drawn jointly with $Y$. Thus we can apply the ``single-sample" bounds to $I(\mathcal{X}_{1:K}:\mathcal{Y})$.
	
	For example, for the NWJ bound Eq.~(\ref{NWJ_bound}) the optimal \emph{ansatz} for the critic $f$ is given by (see Sec.~\ref{NWJ_extremum} in the Supplemental Material):
	\begin{equation}
		f^*(x_{1:K},y)=1+\log \frac{p(y|x_{1:L})}{p(y)} = 1 + \log\frac{p(y|x_1)}{p(y)}.
	\end{equation}
  This critic function can be made to take advantage of the additional replicas. Observing that by Eq.~(\ref{EBA_ap}):
	\begin{equation*}
		\frac{p(x_1|y)}{p(x_1)}=\frac{e^{f^*(x_1,y)}}{Z(y)},
	\end{equation*} 
	we can take the following modified critic function:
	\begin{align}\label{NCEansatz}
		g(x_1,y) &:= 1 + \log\frac{e^{f(x_1,y)}}{m(y;x_{1:K})}
	\end{align}
	where $m(y;x_{1:K})$ is the $K$ sample Monte Carlo estimate of the partition function $Z(y)$:
	\begin{equation}\label{eq:si_mzy}
		m(y;x_{1:K}) := \frac{1}{K}\sum_{i=1}^Ke^{f(x_i,y)} \approx Z(y).
	\end{equation}
	
	Substituting this critic in the NWJ bound for $I(\mathcal{X}_1:\mathcal{Y})$, and averaging over $K$ replica random variables such that $(\mathcal{X}_i, \mathcal{Y}_i)\sim p(x_i,y_i)$, \emph{i.e.}~using the NWJ bound with each $\mathcal{Y}_j$ playing the role of $\mathcal{Y}$ in turn, we arrive after some simple but tedious algebra at the InfoNCE lower-bound of MI:\cite{poole2019variational}
	\begin{align}\label{InfoNCE_end}
		I(\mathcal{X}:\mathcal{Y})\geq& I_{\rm NCE}(\mathcal{X}:\mathcal{Y}):=\langle I_{\rm NWJ}(\mathcal{X}:\mathcal{Y}) \rangle\\
		&=\frac{1}{K} \mathbb{E}_{\prod_{k=1}^K p(x_k, y_k)}  \left[\sum_{j=1}^K \log \frac{e^{f(x_j,y_j)}}{\frac{1}{K}\sum_{i=1}^{K}e^{f(x_i,y_j)}}\right].\nonumber 
	\end{align}
	This completes derivation of the InfoNCE lower-bound for MI, which is the default one used in our implementation. In Sec.\ref{NCE_upper} in the Supplemental Material we discuss some of its further properties, including the expression in terms of the categorical cross-entropy and conditions on its upper bound, which should be taken into account to avoid biased estimates.
	
	\subsubsection{Neural network architectures for the RSMI lower-bound estimator}
	A key idea which made the variational bounds for MI introduced above computationally relevant was to
	parametrize the critic functions by neural networks $f_{\Theta}$.\cite{belghazi2018mine} Their parameters $\Theta$ can then be optimized using standard methods, \emph{e.g.}~stochastic gradient descent, to maximize the lower bounds.
	
	Multiple multi-layer perceptron (MLP) architectures for the critic function $f\equiv f_\Theta(h,e)$ have been considered.\cite{poole2019variational} Here, we opt for a \textit{separable} form, such that: 
	\begin{equation}
		f_\Theta(h,e)=v^{\rm T}(h)u(e), 
	\end{equation}
	where $v$ and $u$ are array-valued functions (here, neural networks, whose weights constitute $\Theta$) that depend only on hidden variables and the environment, respectively. The networks $v$ and $u$ independently map $\mathcal{H}$ and $\mathcal{E}$ to a so-called embedding space. This choice allows construct the scores matrix $F_{ij}$ (see below), storing the values of $f_\Theta$ for all pairs of jointly and independently drawn samples, in $N$ passes of the MLP ($N$ passes for both $v$ and $u$ networks) for a sample dataset of size $N$. This is in contrast to $N^2$ passes for all $N(N-1)$ independent and $N$ joint samples in a concatenated architecture $f_\Theta(h,e)=f_\Theta([h,e])$.
	We opted for two hidden layers each with 32 neurons fully-connected to the layer containing the $(\mathcal{H},\mathcal{E})$ data. The embedding dimension is 8. The neurons are activated by the rectified linear unit (ReLU) function (see,\textit{ e.g.}~Ref.~\onlinecite{Goodfellow-et-al-2016}). We note that the results of RSMI-NE are not sensitive to these architectural details.

	\subsection{The coarse-graining network}\label{cg_net}
	
	The second key element of the algorithm is the coarse-graining probability distribution $p_\Lambda(h|v)$. To take advantage of the differentiable nature of the RSMI estimators described above, and the possibility of efficient gradient descent training, we consider \emph{ansätze} parametrised by neural networks, as well.
	In particular we use the following composite architecture (see Fig.~\ref{fig:netfig}): 
	\begin{equation}
		h = \tau \circ (\Lambda \cdot v).
	\end{equation}
	Combinations of the local DOFs are selected by an inner product with the parameters $\Lambda$, which can be understood in terms of a generalised Kadanoff block-spin transformation,\cite{PhysicsPhysiqueFizika.2.263} before being mapped to a discrete variable by the map $\tau$. In practice the first operation can be represented by a single layer network with parameters $\Lambda$, and the number of kernels can be varied according to the symmetries of the system. We emphasize that the RSMI approach does not rely on the specific type of the variational \emph{ansatz} for coarse-graining; the inner product form is a choice of convenience here. We briefly discuss the possibility of a more general coarse-graining $\Lambda$ network \emph{ansatz}, comprising multiple layers, in Subsection~\ref{subsection:num_components}.
	
\begin{algorithm*}[]
	\caption{One epoch for the unsupervised learning procedure for the RSMI-net using InfoNCE lower-bound}
	\label{RSMI_alg_training}
	\begin{algorithmic}[1]
		\State $\eta=$ learning rate
		\State $\epsilon=$ relaxation parameter for Gumbel-softmax distribution
		\State $\Theta^0 \leftarrow $ random hyperparameter tensor \Comment{initialise InfoNCE \textit{ansatz} $f(h,e)$}
		\State $\Lambda^0 \leftarrow $ random hyperparameter tensor \Comment{initialise coarse-graining filter}
		\For {$s$ in $1:n$ }\Comment{loop over all $n$ $K$-replica samples for $(\mathcal{V},\mathcal{E})$}
		\State  $\epsilon^s \leftarrow$ reduce Gumbel-softmax relaxation parameter
		\State $\tau^s \leftarrow \tau(\epsilon^s)$ \Comment{Anneal Gumbel-softmax layer}
		\For {$i$ in $1:K$}
		\For {$j$ in $1:K$}
		\State $h_i^s[\Lambda^s] \leftarrow \tau^s (\Lambda^s\cdot v_i^s)$ \Comment{Coarse-grain visible degrees of freedom}
		\State $F_{ij}(\Theta^s, \Lambda^s)\leftarrow f(h_i^s[\Lambda^s],e_j^s;\Theta^s)$ \Comment{$ij$'th element of scores matrix}
		\EndFor
		\EndFor
		\State $Q(x_{1:K},y_{1:K};\Theta^s, \Lambda^s) \leftarrow \sum_{j=1}^K \frac{F_{jj}(\Theta^s, \Lambda^s) }{\sum_{ij=1}^{K}\exp{F_{ij}(\Theta^s, \Lambda^s)}}$ \Comment{InfoNCE "prediction"}
		\State \Comment{Update parameters of the RSMI estimator network:}
		\State $\Delta \Theta^s\leftarrow \eta \nabla_{\Theta}\left[\log Q(\Theta,\Lambda^s)\right]\Big{|}_{\Theta=\Theta^s}$ \Comment{automatic differentiation}
		\State $\Theta^s\leftarrow \Theta^s + \Delta \Theta^s$ \Comment{stochastic gradient-ascent}
		\State \Comment{Update parameters of the coarse-grainer network:}
		\State $\Delta \Lambda^s\leftarrow \eta \nabla_{\Lambda}\left[\log Q(\Theta^s, \Lambda)\right]\Big{|}_{\Lambda=\Lambda^s}$
		\State $\Lambda^s\leftarrow \Lambda^s + \Delta \Lambda^s$ 
		\EndFor
		\State $\tilde{I}_{\Lambda}(\mathcal{H}:\mathcal{E})=\frac{1}{n}\sum_{t=1}^n \log Q(x_{1:K},y_{1:K};\Theta^t,\Lambda^t) + \log K$ \Comment{average over $n$ samples}
		\State\Return  $\tilde{I}_{\Lambda}(\mathcal{H}:\mathcal{E})$, $\Lambda^{n}$
	\end{algorithmic}
\end{algorithm*}
	
	The final step is a non-linear stochastic mapping $\tau$ into a state $h$ of the coarse-grained variable with a pre-determined type (\emph{e.g.}~binary spins). This embedding is both crucial,\cite{PhysRevX.10.031056} and algorithmically non-trivial, as the discretisation operation needs to be differentiable.\cite{jang2016categorical}

	\subsubsection{Gumbel-softmax reparametrisation trick for discretisation of coarse-grained variables}
	In the RSMI-NE the coarse-grained variables $h$ are inputs to the MI estimator. Since the value of MI depends on what kind of distribution $h$ belongs to, we need to ensure that this estimation step is not falsified by \emph{e.g.}~neglecting to force the output of the coarse-grainer into a discrete binary variable form, rather than a real number, if we decided $h$ to be Ising spins. The apparent problem is that generating stochastic discrete $h$ seems to spoil the differentiability of the whole setup. This is in fact somewhat similar to the problem encountered in variational autoencoders (VAEs), which is solved there using the so-called \emph{reparametrization trick}, effectively allowing to only differentiate w.r.t.~to the parameters of the latent space probability distribution. With this intuition in mind, we discuss the solution to the issue in RSMI-NE.\cite{NIPS2014_5449}
	
	The solution has three steps. The first result needed is the \emph{Gumbel-max reparametrisation}: let $h$ be a categorical random variable\index{categorical random variable} which can be in one of the states $\{i\}_{i=1}^N$ with the set of probabilities $\{\pi_i\}_{i=1}^N$. It can be shown that:
	\begin{equation}\label{gumbel_max}
		k^*=\argmax_{k\in\{1:N\}}\left\{g_i + \log \pi_i\right\}_{i=1}^N 
	\end{equation}
	is a categorical random sample drawn from the distribution defined by $\{\pi_i\}_{i=1}^N$, where $\{g_i\}_{i=1}^N$ are $N$ \textit{parameterless} random variables drawn from the Gumbel distribution \index{Gumbel distribution} \cite{gumbel1954,NIPS2014_5449} centred at the origin. All the parametric dependence is therefore in the constants $\{\pi_i\}_{i=1}^N$, separated from the source of randomness, which in principle allows differentiation of the distribution. 
	
	Since $\argmax$ is not differentiable itself, in the second step it is \emph{smoothened}, in a controlled and reversible fashion. Given $\{g_i\}_{i=1}^N$ we define a vector-valued random variable utilizing the softmax function Eq.~(\ref{eq:si_softmax}), whose $j$-th component takes the form:
	\begin{equation}
		{\rm softmax}_{j,\epsilon}\left(\{g_i + \log \pi_i\}_{i=1}^N\right)=\frac{\exp\left[(\log \pi_j + g_j)/\epsilon\right]}{\sum_{i=1}^N\exp\left[(\log \pi_i + g_i)/\epsilon\right]},
	\end{equation}
	where $\epsilon$ is the smearing parameter. For $\epsilon \rightarrow 0$ the softmax becomes the argmax, mapping the argument vector $y=\{g_i + \log \pi_i\}_{i=1}^N$ into a $N$-component one-hot vector (one-hot encoding maps each of $N$ possible states $i$ of a discrete variable into a $N$-dimensional vector, with $1$ on $i$-th position, and zeros elsewhere) with some $k^*$-th entry taking the value $1$, thereby marking $y_{k^*}=\max y$. 
	The result is a Gumbel-softmax random variable;\cite{NIPS2014_5449} approximately (or pseudo-) discrete, for small enough $\epsilon$ (do not confuse with a discrete random variable defined by taking the maximum component of the softmax function). 
	For $\epsilon\approx0$, a sample vector $h\sim{\rm softmax}_\epsilon(\{g_i + \log \pi_i\}_{i=1}^N)$ has a single component very close to $1$ and all other components take very small values, comparable to machine precision. Conversely, in the limit $\epsilon\to \infty$ the distribution becomes uniform over all components.
	
	\begin{figure*}
		\centering
		\includegraphics[width=1\linewidth]{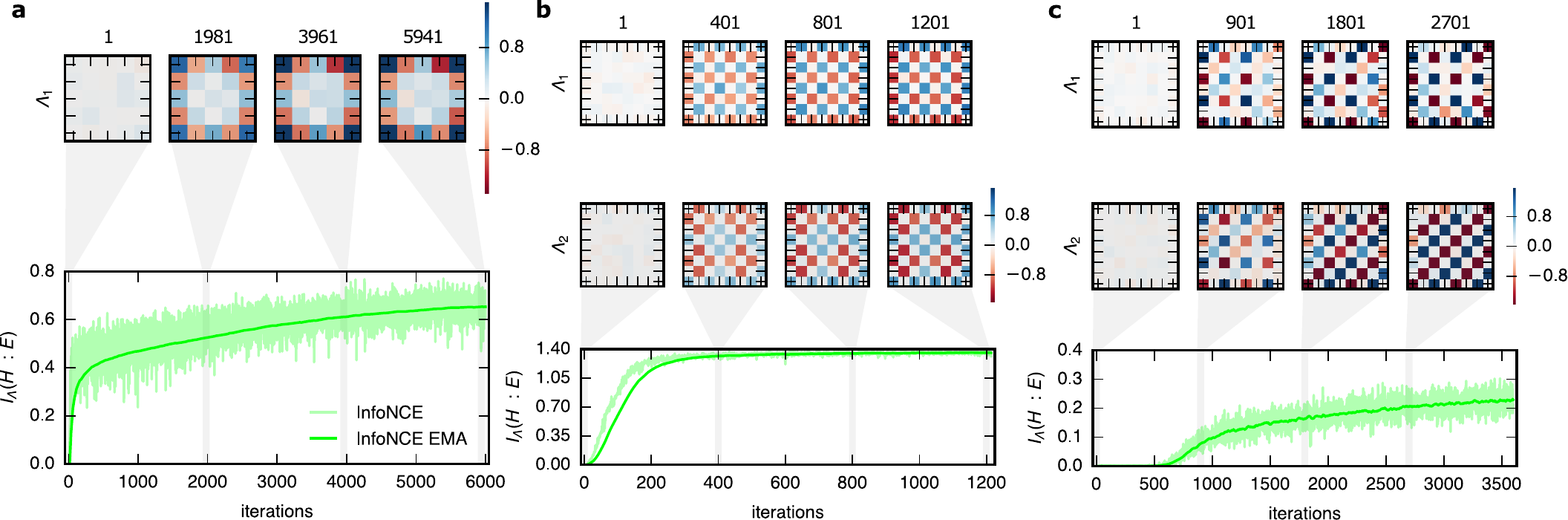}
		\caption{{\bf Convergence of real-space mutual information value and the coarse-graining filters.} The light green curve shows the the time series of RSMI, and dark green its exponential moving average. In the top panel, the time series of the coarse-graining filters are given. {\bf a} The antiferromagnetic Ising model on a 2D square lattice at the critical point.  The RSMI converges to $\log2$ and the optimal filter couples to the boundary degrees of freedom in $\mathcal{V}$ with an alternating sign pattern, due to the onset of anti-ferromagnetic order. {\bf b} The same for the interacting dimer model at $T=0.4<T_{\ssm BKT}$ {\bf c} For the interacting dimer model at $T=15.0\gg T_{\ssm BKT}$. See Sec. \ref{equilibrium} for details on the Ising and dimer models, and the interpretation of these results.}
		\label{fig:rsmi_conv}
	\end{figure*}	
		
	This is used in the third step, where we anneal the smoothing parameter. There is a trade-off between small $\epsilon$ which leads to very noisy gradient estimates, and large $\epsilon$ at which the gradients have low variance but the samples $h$ are far from being discrete. To reconcile this, we start the training at a high value of $\epsilon$ and anneal it exponentially towards a small positive value during training and thus stiffen the pseudo-discrete variable into an increasingly better approximation of a discrete one. The annealing procedure is described in more detail in Sec.~\ref{GS_annealing} in the Supplemental Material.

	\subsection{Unsupervised learning scheme for the combined network}\label{RSMI-net_training}
	
	The results of the preceding section enable us to construct a variational \textit{ansatz} $\tilde{I}_{\Lambda, \Theta}(\mathcal{H}:\mathcal{E})$, differentiable with respect to the parameters of the coarse-graining filter $\Lambda$ and the estimator $\Theta$. We stress that it is upper-bounded by the exact value of RSMI:
	\begin{equation}
		\max_\Lambda  \tilde{I}_{\Lambda, \Theta}(\mathcal{H}:\mathcal{E}) \leq 	I_{\Lambda^*}(\mathcal{H}:\mathcal{E}), \hspace{0.5 cm} \forall \Theta,
	\end{equation}
	where $\Lambda^*$ stands for the optimal solution. The equality holds if and only if the estimator becomes exact, \emph{i.e.} for the optimal parameters $\Theta=\Theta^*$ of the energy based \emph{ansatz} $f$ of InfoNCE. Thus, the search for the optimal RSMI coarse-graining, for any tuning parameters of the system, becomes a well-defined and tractable variational problem. It can be solved by simultaneously optimising both set of trainable parameters $\{\Lambda, \Theta\}$ towards the same objective in an unsupervised learning scheme, which we now describe.
	
	\begin{figure*}[]
		\centering
		\includegraphics[width=1\linewidth]{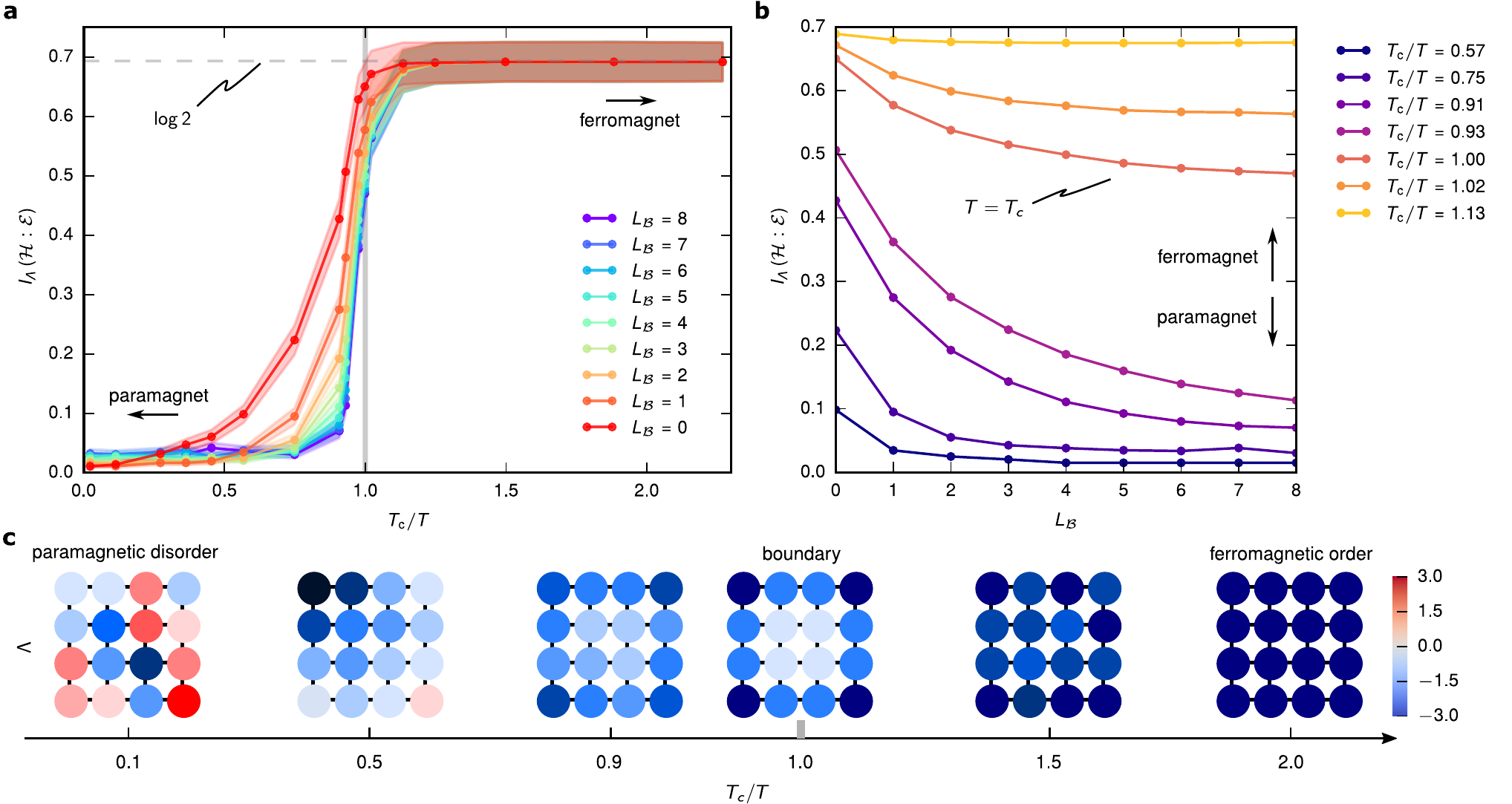}
		\caption{{\bf Maximal RSMI as a function of temperature and its scaling with $L_\mathcal{B}$.} {\bf a} The dependence of the maximal RSMI on temperature for different buffer widths $L_{\mathcal B}$. {\bf b} The scaling of the maximal RSMI with $L_{\mathcal B}$ at different temperatures. It is found that the RSMI decays exponentially in the paramagnetic phase, whereas the decay is slower at $T\leq T_c$. {\bf c} The evolution of the RSMI-optimal filters with temperature at $L_{\mathcal B}=4$.}
		\label{fig:isingrsmi}
	\end{figure*}
	
	The inputs of the RSMI-NE can be \emph{e.g.}~the Monte Carlo (MC) samples from the desired model, for example as in Sec.~\ref{equilibrium}, but the algorithm can also be run on measured data. Since we estimate RSMI using the InfoNCE bound, the sampling is divided into mini-batches, each containing $K$ samples. We separate in each sample the visible patch $\mathcal{V}$ and its environment $\mathcal{E}$, dismissing a finite buffer that separates them. Then a single mini-batch is denoted by the multi-dimensional random variable $(v_{1:K},e_{1:K})=(v_1,\cdots,v_K,e_1,\cdots,e_K)$. As usual, ensuring good quality MC sampling is important.
	
	Let $\Lambda^s$ and $\Theta^s$ denote the network parameters for the coarse-graining, and the critic $f$, respectively, at training step $s$. We initialise them as tensors containing random numbers. At each step $s$, in the samples in the mini-batch $v_i$ are coarse-grained into $h_i[\Lambda^s]$ and the scores matrix $F_{ij}(\Theta^s, \Lambda^s)=f(h_i[\Lambda^s], e_j;\Theta^s)$ is computed for the InfoNCE at current values of the network parameters. In $F_{ij}$ the entries with $i=j$ denote the jointly drawn samples and the rest denote independently drawn samples for the coarse-grained degree of freedom and the environment.
	As described above discrete $h$ are generated by a layer $\tau$.
	
	The InfoNCE prediction [that of $p(h,e)$ being equal to $p(h)p(e)$ or not, as defined in Eq.~(\ref{NCE_expectant})] for the mini-batch is computed using the scores matrix as:
	\begin{equation}\label{NCE_RSMI_prediction}
		Q(h_{1:K},e_{1:K};\Theta^s, \Lambda^s) = \sum_{j=1}^K \frac{\exp F_{jj}(\Theta^s, \Lambda^s) }{\sum_{i=1}^{K}\exp{F_{ij}(\Theta^s, \Lambda^s)}}.
	\end{equation}
	Then $\log Q(h_{1:K},e_{1:K};\Theta^s, \Lambda^s) + \log K$
	gives our single mini-batch estimate of RSMI.
	
	The gradients of the mini-batch estimate of RSMI with respect to $\Lambda$ and $\Theta$ are used to update the network parameters. We use the adam optimiser\cite{kingma2014adam} to perform stochastic gradient-ascent. We found that using the same learning rate for both $\Lambda$ and $\Theta$ leads to efficient training. We repeat the above over all mini-batches, until all samples are fed to the network once. This constitutes one epoch of training. In Alg.~\ref{RSMI_alg_training} the training procedure is given in pseudo-code.
	
	We train for multiple epochs until convergence criteria are satisfied (see Sec.~\ref{convergence} in the Supplemental Material). For illustration, we plot in Fig.~\ref{fig:rsmi_conv} the time series of the RSMI estimates and the coarse-graining filters during the training for 2D critical Ising anti-ferromagnet, and interacting dimer models below and above the BKT transition point (see Sec.~\ref{equilibrium}). 
	Upon convergence, we are left with an optimised coarse-graining represented by the final $\Lambda$-parameters, and an estimate of the RSMI given by a moving average of the time-series of mini-batch estimates. 
	
	\section{RSMI-NE in equilibrium systems}\label{equilibrium}	
	The RSMI-NE algorithm yields a comprehensive characterization of long-distance properties of an equilibrium statistical system: its phase diagram, correlations, symmetries. The companion work\cite{rsmine_letter} discusses in detail the construction of order parameters, or, more generally, relevant operators. Here we demonstrate how the extracted quantities, and their dependence on the tuning parameters of the system and the buffer length-scale, reveal the critical points and the nature of correlations in the phases. We illustrate this on the examples of a dimer model with aligning interactions and the 2D Ising model. We examine the information, particularly on symmetries, also \emph{emergent}, contained in the statistical \emph{ensemble} of coarse-graining filters, and show its retrieval with ML techniques, which is of practical importance when faced with incomplete inputs. 
	
	We emphasize that, in contrast to many applications of ML (see Refs.\onlinecite{RevModPhys.91.045002,doi:10.1080/23746149.2020.1797528} for a recent review), the success of RSMI-NE in extracting physical data \emph{is not} serendipitous or due to a particular choice of architecture, but a consequence of RSMI being a well-defined \emph{physical quantity},\cite{Gordon2020}, which the ML methods used approximate numerically.

	\subsection{The phase diagram from the parameter dependence of RSMI  and its scaling with buffer size}
	Real-space mutual information quantifies the totality of spatial correlations in the system, and thus their changing structure, especially due to phase transitions, should be reflected in its value. This is indeed the case, as shown below. The nature of these correlations (power-law vs.~exponential) further determines the decay properties of RSMI as a function of the length scale set by the buffer width.
	
	
	We use the example of the classical 2D Ising model: 
	\begin{equation}
		K[x=\{x_i\}]=\beta J \sum_{\langle i ,j\rangle} x_i x_j,
	\end{equation}
	with $x_i=\pm 1$, as the simplest test case for RSMI-NE. It undergoes a second order phase transition between a ferromagnetic for $J<0$ or anti-ferromagnetic order for $J>0$, and a disordered paramagnetic phase at inverse temperature $\beta =\ln(1+\sqrt{2})/2\approx 0.44$.\cite{PhysRev.65.117} We investigate this model in the temperature range $T_c/T\in [0,2.5]$ by optimising RSMI at buffer widths $L_{\mathcal B}\in[0,8]$. 
	
	As shown in Fig.~\ref{fig:isingrsmi}.a, the temperature dependence of the maximal information $I_\Lambda(T)$, \emph{i.e.}~the amount of long-range information attained \emph{with the (at given T) optimal} $\Lambda$, is a clear indicator of the second order phase transition, and of the existence of two phases. At $T<T_c$, independent of the buffer width, exactly 1 bit of information is recovered. This precise quantization is due to RSMI effectively counting the (two) segregated phase space sectors corresponding to the ferromagnetic ground states, and reveals the long-range order.
	
	Phase transitions are reflected by non-analyticities in $I_{\Lambda}(T)$ (\emph{cf.} the behaviour of the mutual information in the absence of buffer in Refs.~\onlinecite{Wilms_2011, PhysRevE.87.022128}). At $T=T_c$ we find that the RSMI has a step-like decay which becomes sharper at larger buffer width $L_{\mathcal B}$. At larger temperatures, the long-range order is destroyed by the thermal fluctuations. The short-range nature of the paramagnetic phase results in an exponential decay of the RSMI with $L_{\mathcal B}$, see Figs.~\ref{fig:isingrsmi}.b and \ref{fig:freedimerrsmiscaling}.a. This is to be contrasted with the critical phase of the dimer model with power-law correlations, where the maximal information decays only algebraically with $L_{\mathcal B}$, see Fig.~\ref{fig:freedimerrsmiscaling}.b.
	
	We next turn to the more complex example of the interacting dimer model, defined by the partition function:
	\begin{equation}\label{eq:dim_part}
		Z(T)=\sum_{\{C\}}\exp{(-E_C/T)},
	\end{equation}
	with $T$ the temperature and $C$ denoting dimer configurations on the square lattice obeying the constraint of exactly one dimer at every vertex, see Fig.~\ref{fig:dimergs}.a. The energy $E_C=N_C(||)+N_C(=)$ counts plaquettes covered by parallel dimers favoured by the interaction.
	
	The essence of this system is in the interplay of aligning interaction energy and entropic effects due to the non-local cooperation of local dimer covering constraints. At low-$T$, the former facilitates long-range order (LRO), crystallizing the system into one of four translation symmetry breaking \textit{columnar} states, see Fig.~\ref{fig:dimergs}.b. With increasing $T$ the system undergoes a Berezinskii-Kosterlitz-Thouless (BKT) transition at $T_{\ssm BKT}= 0.65(1)$,\cite{PhysRevE.74.041124} entering a critical phase characterised by algebraic decay of correlations with exponents continuously changing with $T$. The effective theory of the system is given by a sine-Gordon field theory.\cite{PhysRevLett.94.235702,fradkin_2013, PhysRevE.74.041124} In particular, for $T\to \infty$ the aligning interactions are irrelevant and this description reduces to a free Gaussian field theory.
	
	\begin{figure}
		\centering
		\includegraphics[width=1\linewidth]{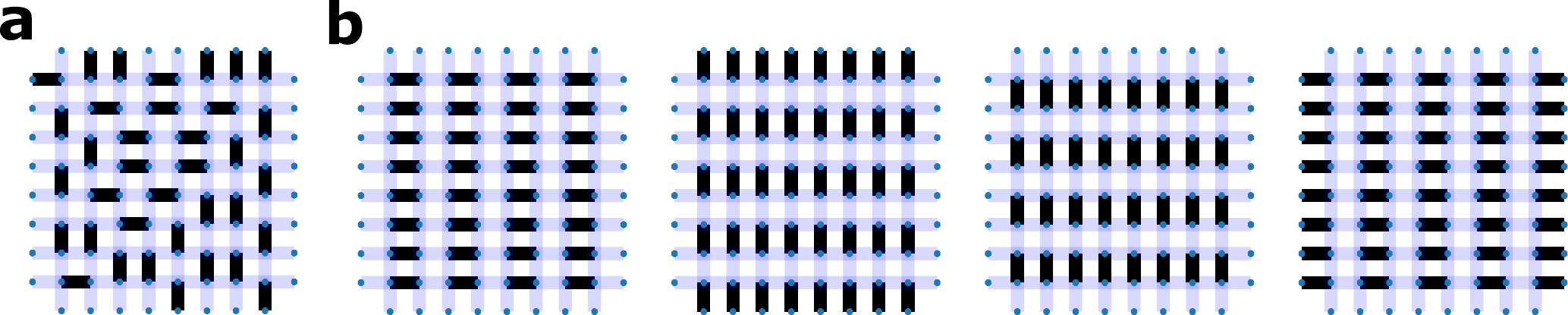}
		\caption{{\bf a} A generic valid dimer covering on the square lattice. {\bf b} The four ground states of the dimer model with aligning nearest neighbour interactions break ${\textsf C}_4$ and lattice translation symmetries.}
		\label{fig:dimergs}
	\end{figure}

	To test our method on the dimer model, we generate its Monte Carlo samples across the whole temperature range, using the directed loop algorithm\cite{PhysRevE.74.041124} (see Sec.~\ref{directedloopMC} in the Supplemental Material for implementation details) for $64\times64$ systems. These are used as inputs to RSMI-NE.
	We restrict the coarse-grained variables $\mathcal{H}$ to a two-component binary vector $\{\pm 1,\pm 1\}$, a choice suggested by the systematic procedure in Sec.\ref{equilibrium}D. Hence, we are looking for a two-component vector of filters $\Lambda_1$,  $\Lambda_2$ determining how the visible region $\mathcal V$ is mapped onto $\mathcal H$. 
	
	\begin{figure}
		\centering
		\includegraphics[width=1\linewidth]{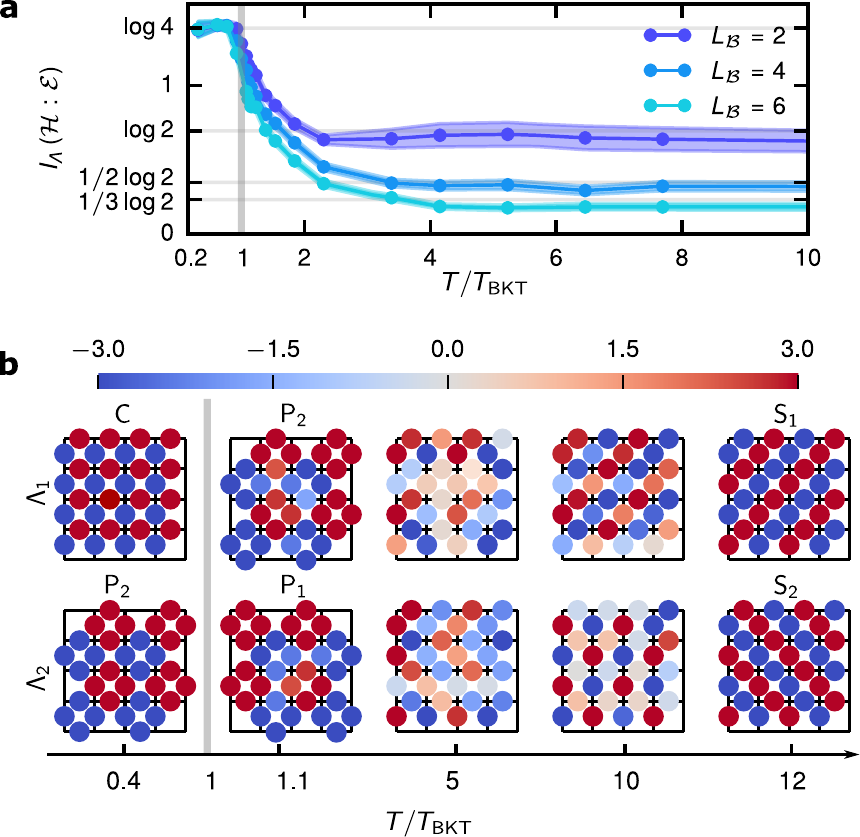}
		\caption{{\bf RSMI analysis of the interacting dimer model.} {\bf a} Total RSMI extracted with the optimal filters as a function of $T$ and its scaling with the buffer size. {\bf b} Samples of optimal filters obtained with RSMI-NE for different $T$ [columnar $({\textsf C})$, plaquette $({\textsf P}_1, {\textsf P}_2)$ and staggered $({\textsf S}_1, {\textsf S}_2)$].}
		\label{fig:dimerfig}
	\end{figure}

	Though the BKT transition is of entirely different nature to the Ising example considered above,
	we find that optimizing the filters $\Lambda_1$,  $\Lambda_2$ for all $T$ readily reveals the structure of the phase diagram (see Fig.~\ref{fig:dimerfig}.a). To wit, for $T<T_{\ssm BKT}$ its value is constant and equal to $\log 4$, or $2$ bits. The information shared between distant parts of the system in the ordered phase is precisely which of the four columnar states they are in. This is analogous to the ferromagnetic order of the Ising model, \emph{i.e.}~the optimal RSMI counts the number of segregated phase space sectors in long-range ordered phases. Moreover, the algebraic decay of $I_{\Lambda}(T)$ with the buffer size for $T>T_{\rm BKT}$, as seen in Fig.~\ref{fig:freedimerrsmiscaling}, is indicative of a critical phase with power-law decaying correlations. In particular, we have found that the RSMI scales as $I_\Lambda(T\to\infty)\sim L_\mathcal{B}^{-\upsilon}$, with the exponent obtained from the best fit (see Fig.~\ref{fig:freedimerrsmiscaling}.c) to be $\upsilon\approx 1.16$.
	
	\begin{figure}
		\centering
		\includegraphics[width=1\linewidth]{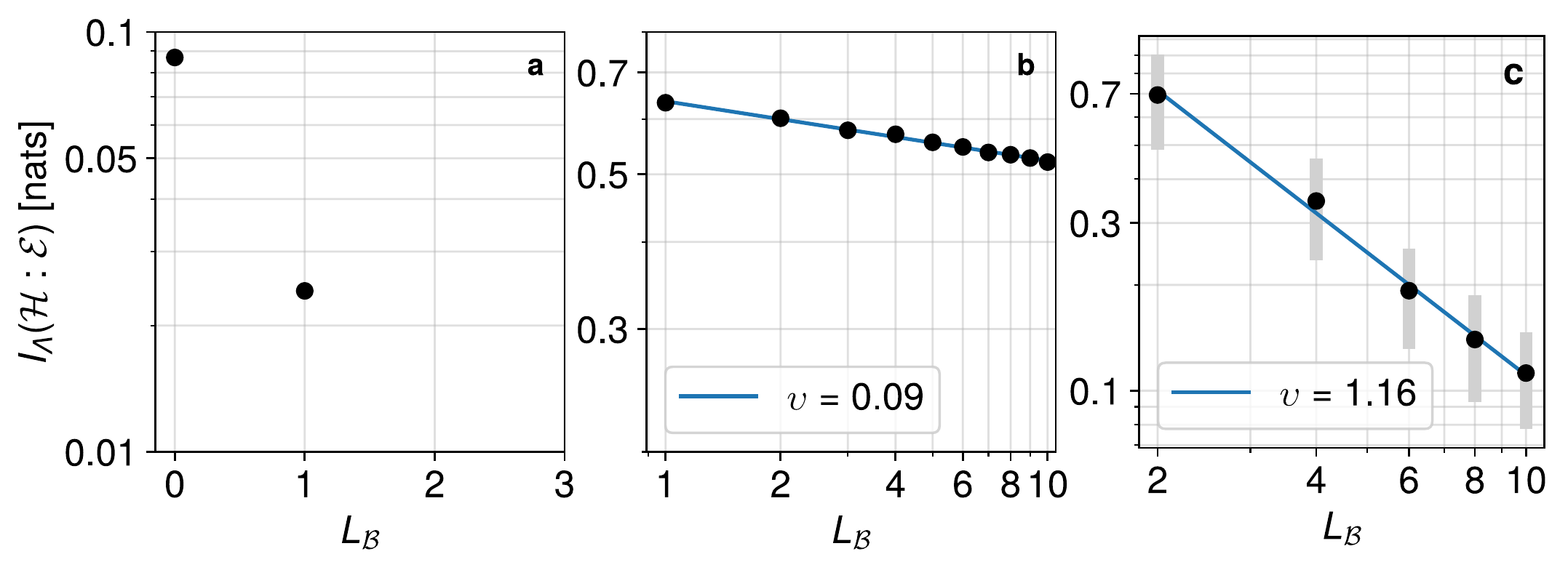}
		\caption{ {\bf The scaling of RSMI with the buffer size $L_\mathcal{B}$.}  {\bf a} The RSMI decays exponentially in the Ising paramagnet (at $T=4$). For $L_{\mathcal B}\geq 2$, the RSMI numerically drops to $0$. {\bf b} At the Ising critical point, the RSMI scales algebraically; the line gives the best fit for the exponent $\upsilon\approx 0.09$.  {\bf c} Same as in {\bf b} but for the high-temperature free dimers ($T=15$). The fit gives the RSMI scaling exponent $\upsilon\approx 1.16$ for the free dimers.}
		\label{fig:freedimerrsmiscaling}
	\end{figure}

	We conclude that the value of the optimal RSMI as a function of the system's parameters provides us with information about position of the critical points, type of phase transition, the nature of correlation decay in the phases, as well as the number of sectors in the long-range ordered phases.
	
	\subsection{Correlations from the parameter dependence and flow of the optimal coarse-graining filters}
	
	Much more can be learnt about spatial correlations upon examining the coarse-graining filters $\Lambda(T)$. First, the optimal filters, with which the highest RSMI value was attained, themselves depend on the tuning parameters of the physical system, and in fact carry the information about the phase diagram. Particularly, they reflect the symmetries of the system (see also Sec.\ref{equilibrium}C). Moreover, the filters depend on the length scale $L_{\mathcal B}$, reflecting an RG flow. In Ref.~\onlinecite{rsmine_letter} we further show that they correspond to the (lattice discretisation of) relevant operators in the field theory describing the system, in light of which observation the intriguing results of this subsection become natural.

	\subsubsection{The optimal coarse-graining filters of the 2D Ising model}
	
	The temperature dependence and the relation of the optimal filters to the phase diagram are clear in the Ising model results, as seen in Fig.~\ref{fig:isingrsmi}.c. Here we used a fixed buffer width $L_{\mathcal B}=4$, a visible region of size $4\times 4$ and a $64\times 64$ grid for the whole system. 
	
	In the high- and low-$T$ limits the paramagnetic and ferromagnetic phases result in optimal filters, which are respectively random and uniform. The uniform filter acts as the ferromagnetic order parameter, labelling the configurations by their magnetisation. Consistently, for the antiferromagnetic Ising model, we found that the optimal filter at low-$T$ is the staggered magnetisation (see Fig.~\ref{fig:rsmi_conv}). At the second-order critical point, the filters exhibit a boundary behaviour,\cite{Koch-Janusz2018} \emph{i.e.}~they correspond to the magnetization on the boundary of the visible block. This is because in the \emph{critical} Ising system the information shared between $\mathcal{V}$ and $\mathcal{E}$ is proportional to the surface of their interface, and not to the volume of $\mathcal{V}$.\cite{Wilms_2011, PhysRevLett.100.070502} The boundary filter thus signals the presence of the critical point.

	\begin{figure}[]
		\centering
		\includegraphics[width=\linewidth]{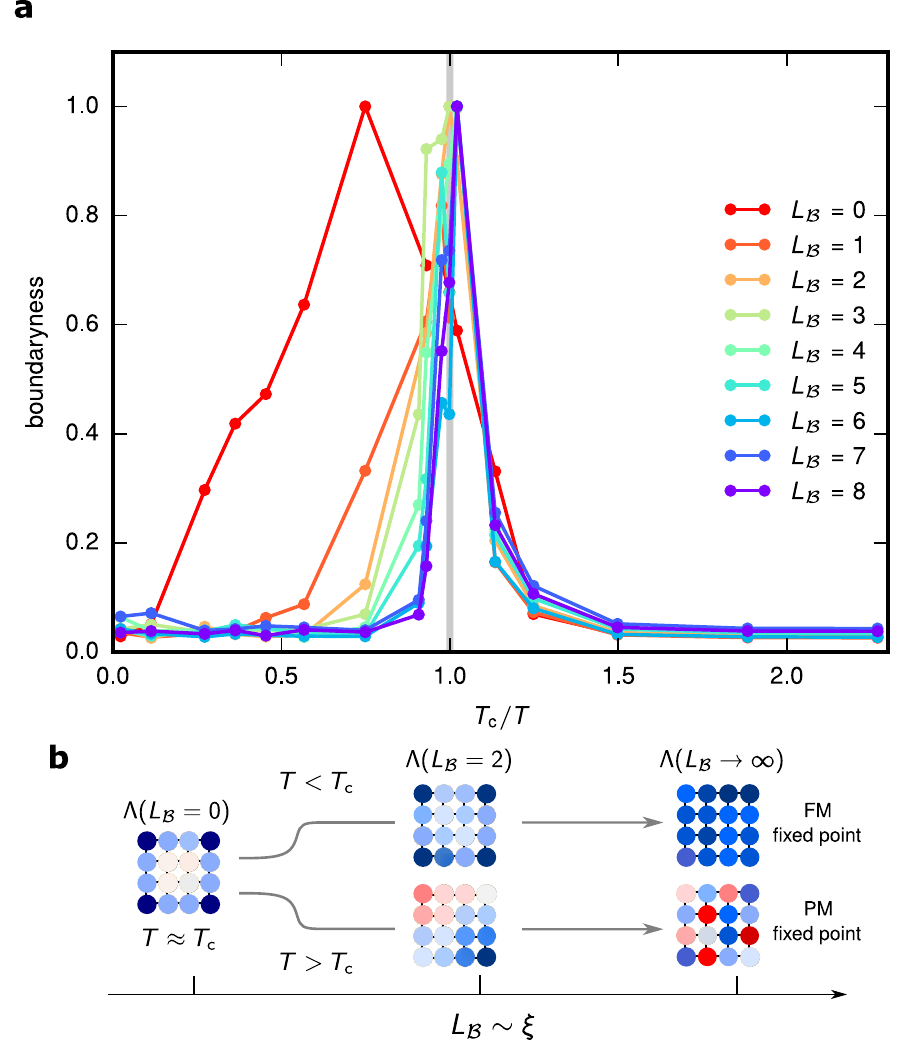}
		\caption{{\bf Relative strength of coupling to the boundary \emph{vs.}~bulk spins in $\mathcal{V}$, as a function of temperature.}  {\bf a} The empirical ``boundaryness" measure (see Eq.~\ref{eq:boundaryness}) of the optimal coarse-graining filter at different temperatures and buffer sizes. {\bf b}  Near the critical point $T_c$ the optimal $\Lambda$ averages the boundary spins in $\mathcal{V}$. The separation $L_\mathcal{B}$ of $\mathcal{V}$ from its environment $\mathcal{E}$ effectively sets the RG scale $\xi$. Growing $L_\mathcal{B}$ increasingly differentiates the filters $\Lambda$ slightly below and above $T_c$, which ultimately flow to the ferromagnetic (FM) and paramagnetic (PM) fixed points. Note the transformations are \emph{not} iterated. 
		}
		\label{fig:isingboundaryness}
	\end{figure}
	
	While the occurrence of the boundary filter is associated with the critical $T_c$, the range of temperatures where this happens depends on $L_{\mathcal B}$. Put differently, the accuracy to which the critical fixed point can be resolved depends on the length scale set by the size of the buffer. To visualize this in Fig.~\ref{fig:isingboundaryness}.a we plot an empirical measure of the relative strength of the boundary vs.~bulk couplings in $\Lambda(T)$, at different values of $L_{\mathcal B}$, which we call  the ``boundaryness":
	\begin{equation}\label{eq:boundaryness}
		{\rm boundaryness} := \frac{|\sum_{i \in {\rm boundary}}\Lambda_i|}{|\sum_{i \in {\rm bulk}} \Lambda_i|},
	\end{equation}
	and we rescale by the maximum over $T$ in Fig.~\ref{fig:isingboundaryness}.a for clarity (note that the maximum value was found approximately the same for all $L_\mathcal{B}$ except for $0$, where it was much greater).
	This ratio peaks around $T=T_c$, becoming increasingly sharp as $L_{\mathcal B}$ grows. 
	
	This behaviour is readily understood, since $L_{\mathcal B}$ (together with the total finite system size, \emph{cf.}~Eq.5 in Ref.\onlinecite{Gordon2020}) effectively controls the RG scale. Indeed, a corresponding flow of the optimal filters can be constructed. As shown in Fig.~\ref{fig:isingboundaryness}.b, for small $L_{\mathcal B}$ at $T\approx T_c$ the optimal filter is a boundary one, both above or below $T_c$. At this scale the critical point is not resolved very well. 
	As $L_{\mathcal B}$ is increased, however, the $T < T_c$ and $T > T_c$ cases are increasingly differentiated, and they eventually flow to the ferromagnetic and paramagnetic fixed points, respectively. This, of course, is consistent with the presence of a repulsive fixed-point in the RG flow of the 2D Ising model.

	\subsubsection{The optimal coarse-graining filters of the interacting dimer model}
	
	The optimal coarse-grainings of the dimer model Eq.~(\ref{eq:dim_part}) likewise
	depend on the tuning parameters of the system (\emph{i.e.}~the temperature), see Fig.~\ref{fig:dimerfig}.b. In contrast to the 2D Ising model example, however, this dependence is continuous for $T > T_{\rm BKT}$. This, in fact, provides another indication that the transition is of the BKT type (in addition to the algebraic scaling of the RSMI curve in the critical phase).
	
	More concretely, in the high- and low-$T$ limits, three classes of filters emerge: independent optimizations (see discussion of the filter \emph{ensemble} in Sec.\ref{equilibrium}C) return exclusively sets of filters $\Lambda_{1,2}$ that correspond to columnar and plaquette at low temperatures, and staggered ones at high temperatures. They are denoted as $\textsf C$, ${\textsf P}_{1,2}$ and ${\textsf S}_{1,2}$ in Fig.~\ref{fig:dimerfig}.b. We call these filters ``pristine'' as they reflect limiting cases. They also reveal information about the symmetries. In particular the pristine plaquette and columnar filters at $T\to 0$ break the discrete translation or rotation symmetry of the lattice, respectively. Any pair of $\Lambda_{1,2}$ drawn out of these classes of filters defines a bijection between the four columnar states in Fig.~\ref{fig:dimergs} and the four distinct states $(\pm 1,\pm 1)$ of the compressed degrees of freedom in $\mathcal H$. They thus label uniquely the ordered states, (which is the reason the recovered RSMI is exactly $2$ bits for $T<T_{\rm BKT}$), 
	and correspond precisely to the dimer symmetry breaking order parameter of Ref.~\onlinecite{PhysRevE.74.041124}. In Ref.~\onlinecite{rsmine_letter} we show the columnar and plaquette filters correspond to the electrical charge operators of the sine-Gordon field theory, \emph{i.e.}~the operators with the lowest scaling dimensions, and so the most relevant in the RG sense.
	
	 The degeneracy of plaquette ${\textsf P}_{1,2}$ and columnar ${\textsf C}$ filters in their RSMI value is lifted when the rotation symmetry is restored at BKT transition: the pristine columnar filter, which breaks the lattice rotation symmetry explicitly, is not found above $T_{\ssm BKT}$. 
	 
	 In the limit of $T\to \infty$ the optimal filers are the staggered ${\textsf S}_{1,2}$. These can be shown\cite{rsmine_letter} to exactly correspond to the spatial gradients of the height field in the sine-Gordon description of the system, or equivalently to the electrical fields. At $T\to \infty$ these are in fact the only terms in the field theory, which is then  that of a free Gaussian field.

	In the critical phase $T>T_{\ssm BKT}$, where the system is characterized by power-law correlations with temperature-dependent exponents, the resulting filters continuously interpolate between pristine plaquette and staggered ones. This is due to the competition between the electric field operator and plaquette correlations, or in other words the gradient and the cosine terms in the sine-Gordon field theory.\cite{fradkin_2013, PhysRevE.74.041124} In a finite system the correlations due to these operators of slightly differing (for $T \gtrapprox T_{\rm BKT}$) scaling dimensions both contribute to RSMI, though in the thermodynamic limit the more relevant gradient term would dominate (which indeed happens for larger $T$).

	\subsection{The ensemble of coarse-graining filters and its analysis: operators and symmetries}
	\subsubsection{Extracting the pristine filters (lattice operators).}
	\begin{figure}
		\centering
		\includegraphics[width=1\linewidth]{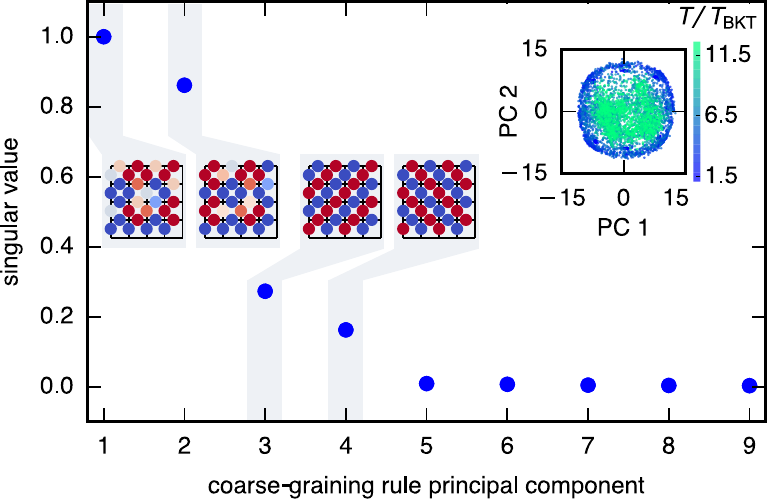}
		\caption{{\bf Analysis of the ensemble of coarse-graining rules.} PCA spectrum of the ensemble of filters $\Lambda(T)$ for a restricted temperature window $0.7 < T< 3.7$ above the BKT point. Top-right inset: projection of the full ensemble on the two highest PCA components. The overlap with those ``plaquette" filters falls with $T$.}
		\label{fig:fig3wblob}
	\end{figure}

	The above described mixing of the pristine filters in intermediate parameter regimes and in finite-size systems may seem troublesome, but we find it can in fact be resolved, and the solution to this problem is useful in itself. 
	
	The key observation is that due to the RSMI-NE being a stochastic algorithm it produces in independent runs a \emph{distribution} of RSMI-optimal transformations (thus Fig.~\ref{fig:dimerfig}.b shows a sample of filters at each $T$). This distribution defines \emph{the ensemble of filters}, a novel concept we introduce.
	
	The ensemble contains physical information, particularly about the symmetries, and, in contrast to individual filters, also \emph{emergent} ones (see below). 
	Crucially, the ensemble also allows to address the problem of filter mixing due to competing correlations. The pristine filters, which correspond to the (lattice representation of) scaling operators,\cite{rsmine_letter} can be identified not only at the limiting temperatures, but also through data analysis of the ensemble in a window of intermediate temperatures.

	To show this we perform a principal component analysis \index{principal component analysis} (PCA) \emph{of the ensemble} of RSMI-optimal filters for an intermediate temperature range $0.7 < T< 3.7$ above the BKT transition, where the pristine components do not explicitly appear. The goal is to find the most distinctive features of the ensemble which vary with the changing system parameters controlling the location in the phase diagram (here: temperature), while filtering out variations due to a specific realisation of statistical noise or the random initial conditions of the training. As we coarse-grain the dimer model using $8\times 8$ two-component filters, we consider a 64-dimensional vector space where each coarse-graining filter component is a point. The input to the PCA consists of the ensemble of coarse-graining filters, flattened into 1D arrays. The resulting principal components are reshaped back into $8\times8$ arrays, so that each defines a coarse-graining. To visualize the results we then project the full space of coarse-graining filters onto the hyperplane given by the most important principal components.

	The results are shown in Fig.~\ref{fig:fig3wblob}. An important observation is that the highest principal components are in fact given by the pristine filters. This justifies describing the filters in intermediate regimes as ``mixtures", as suggested by the intuitive physical picture of the competing correlations. We can thus identify the relevant operators, \emph{i.e.}~the plaquette (electric charge) and staggered (electric field) filters while never seeing data from parameter regimes where they entirely dominate. This is important, as MC simulations may be costly, or we may be dealing with experimental data whose range we do not have full control over.

	\subsubsection{Discovering broken and emergent symmetries.}
	
		\begin{figure*}[ht]
		\includegraphics[width=0.9\linewidth]{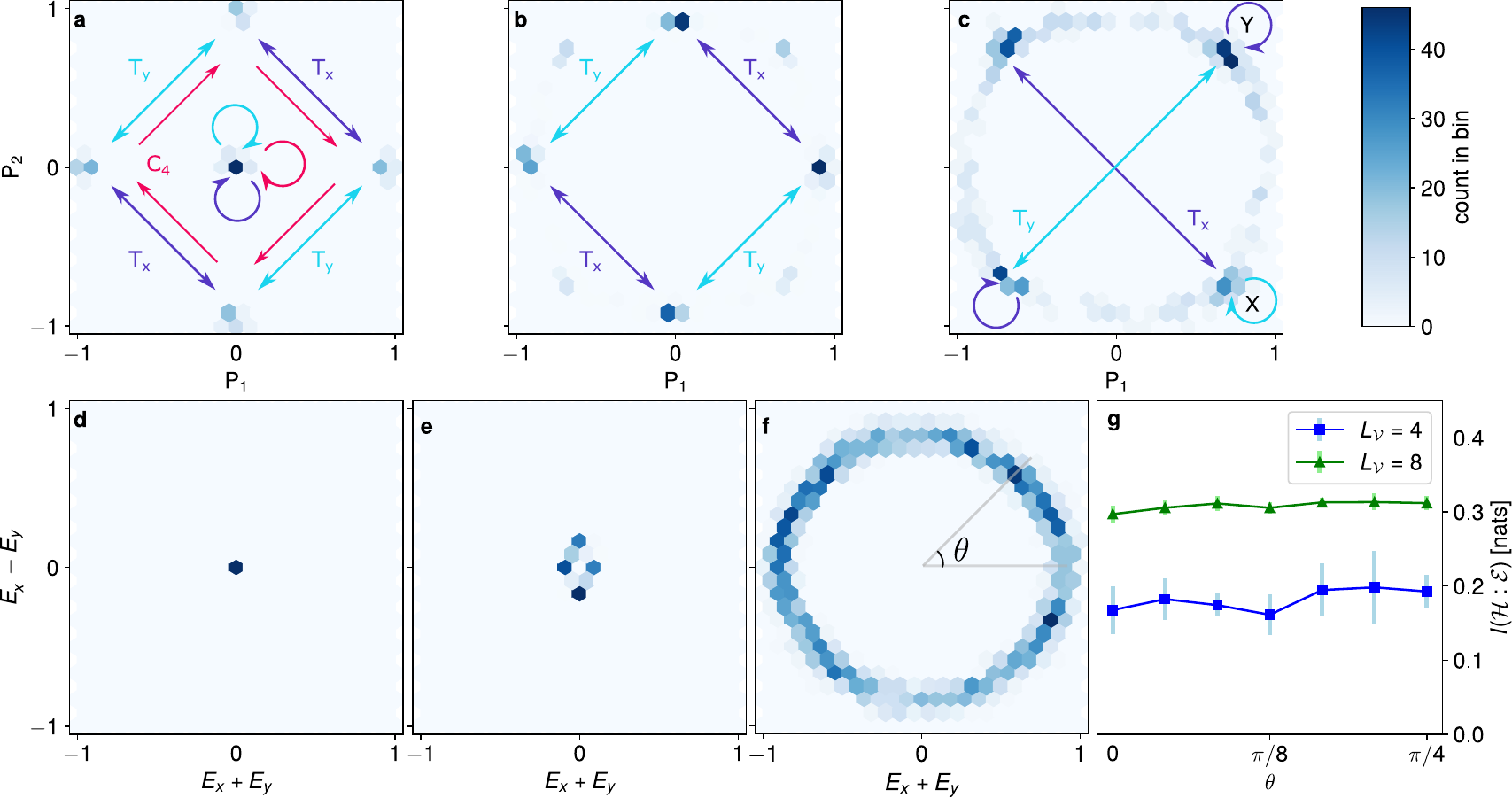}
		\caption{{\bf Identifying symmetries in the ensemble of optimal coarse-grainings.} {\bf a}-{\bf c} The projection of the ensemble of optimal filters onto the orthogonal basis $({\textsf P}_1, {\textsf P}_2)$ of \emph{pristine} plaquette filters at temperatures $\{0.2<T_{\ssm BKT}, 0.7\gtrsim T_{\ssm BKT}, 1.5\}$, respectively. The induced action of the symmetries on the filter components are marked with red arrows for ${\textsf C}_4$ rotation, and with purple and light blue for horizontal and vertical translations (${\textsf T}_x,{\textsf T}_y$), respectively. The central peak in {\bf a} corresponds to the columnar subspace $\pm {\textsf C}$, closed under  the action of both translation and ${\textsf C}_4$.  We see in {\bf b} that there appear two independent representations of the translation symmetry, namely one with the basis $({\textsf P}_1, {\textsf P}_2)$ and also another with $(\textsf{X},\textsf{Y}):=({\textsf P}_1 - {\textsf P}_2, {\textsf P}_1+{\textsf P}_2)$, see discussion in the text. For clarity, we show the ${\textsf T}_{x,y}$ action in these two representations separately in {\bf b} and {\bf c}. {\bf d}-{\bf f} The projection onto the orthogonal basis $(E_x+E_y, E_x-E_y)$ at temperatures $\{0.2<T_{\ssm BKT}, 0.7\gtrsim T_{\ssm BKT}, 15.0\}$, respectively. {\bf f} The ensemble for high-$T$ is approximately invariant under \emph{continuous} rotations in the space of filters. {\bf g} The RSMI is constant as a function of the rotation angle of the electrical field. This degeneracy in RSMI value is the origin of the \emph{emergent} $U(1)$ symmetric ensemble.}
		\label{fig:dimerU1}
	\end{figure*}

	Here we show how the different symmetries of the model, including emergent ones, can be observed in the statistical properties of the filter ensemble. To this end recall first (Sec.\ref{equilibrium}B.2) that a pair of plaquette filters, or a plaquette and a columnar one are degenerate in RSMI for $T < T_{\ssm BKT}$, and label uniquely the symmetry-broken states. This degeneracy is reflected in the equal frequency with which they appear as the optimal solutions in individual RSMI-NE runs. Likewise, the disappearance of the rotation-symmetry-breaking columnar filter from the ensemble above $T_{\ssm BKT}$ signals the lifting of the columnar/plaquette degeneracy and restoration of the rotation symmetry. These observations can be pushed further.
	
	To wit, though each individual instance of RSMI optimization returns an object which we can interpret as an operator, more generally there often exists an RSMI-degenerate subspace for the kernels of these operators, \emph{i.e.}~the filters $\Lambda$. In fact, the ensemble of many optimization instances (at the same physical parameters) allows to map out this subspace, and the resulting distribution contains systematic information about symmetries. To show this, we project the ensemble of filters generated in the dimer model from 500 optimization runs for each temperature onto the orthogonal bases formed by the pristine staggered and plaquette filters introduced above. 
	
	We begin in the ordered phase. In Fig.~\ref{fig:dimerU1}.a, the distribution of the optimal filters [pairs of $\pm{\textsf P}_{1,2}$, or $(\pm{\textsf C},\pm{\textsf P}_{1,2})$] is shown projected onto the plaquette filters, exhibiting four distinct peaks corresponding to $\pm{\textsf P}_{1,2}$, and a central peak corresponding to $\pm{\textsf C}$. It can easily be checked that these form closed subspaces under the action of discrete lattice translations ${\textsf T}_{x,y}$ in $x$ and $y$ directions, and ${\textsf C}_4$ rotations.
	Indeed, since both families of filter pairs are bijections of the four broken-symmetry states, they give simultaneous representations of the two broken symmetries. 
	
	More concretely, the action of ${\textsf C}_4$ around a lattice vertex leads to a $\mathbb{Z}_4$-cycle between the peaks, as shown by the red arrows in Fig.~\ref{fig:dimerU1}.a. Under ${\textsf T}_{x,y}$, on the other hand, $\pm{\textsf P}_{1,2}$ gives a representation of $\mathbb{Z}_2\times \mathbb{Z}_2$, as shown by the light-blue and purple arrows in Fig.~\ref{fig:dimerU1}.a. The columnar subspace $\pm{\textsf C}$ leads to a trivial representation as it is left invariant under both symmetries.
	
	Above the BKT transition, a more complex picture arises.\cite{PhysRevLett.94.235702,PhysRevE.74.041124} Especially in finite systems, the plaquette correlations remain strong and though the ${\textsf C}_4$ symmetry is restored, the low-lying excitations consisting of exchanging horizontally/vertically aligned dimer pairs around plaquettes do not immediately restore the translation symmetry.\cite{PhysRevE.74.041124}
	Since the two broken translation symmetry states differ by the ``site-parity" (\emph{e.g.}~the second and third configurations in Fig.~\ref{fig:dimergs}.b), using one of the coarse-graining components to label the site-parity 
	the filters can still recover at least one bit of long-range information (\emph{cf.}~Fig.~\ref{fig:dimerfig}) in a finite-sized system close to the transition point. 
	
	In particular, the pairs $({\textsf C},{\textsf P}_{1,2})$ become immediately sub-optimal above BKT (\emph{cf.} the disappeared central peak in Fig.~\ref{fig:dimerU1}.b), as they cannot produce a bijective labelling due to the restored ${\textsf C}_4$ symmetry.
	In contrast, the pairs consisting of $\pm{\textsf P}_{1,2}$ can still do the identification faithfully if $T$ sufficiently close to $T_{\ssm BKT}$, so that $\mathcal{V}$ contains at most one flipped plaquette. This is reflected in the persistence of the four $\pm{\textsf P}_{1,2}$ peaks in Fig.~\ref{fig:dimerU1}.b.
	
	Note, however, the development of four additional peaks on the diagonals in Fig.~\ref{fig:dimerU1}.b. They are readily understood: the information-rich ``site-parity" identification is insensitive to multiple plaquette flips if the filters $({\textsf X},{\textsf Y}):=\left(({\textsf P}_1-{\textsf P}_2),({\textsf P}_1+{\textsf P}_2)\right)$\footnote{This corresponds to an alternative (but equivalent) form $(\cos \varphi, \sin \varphi)$ of the $\mathcal{O}_1$ electric charge operator in the continuum sine-Gordon theory with the height field shifted by $\pi/4$ with respect to the expression in terms of ${\textsf P}_{1,2}$.} are used instead. This is because they exclusively couple to horizontal or vertical dimers in each component, hence the notation.
	
	As shown explicitly in Fig.~\ref{fig:dimerU1}.c, ${\textsf T}_x$ (${\textsf T}_y$) inverts the sign of ${\textsf X}$ (${\textsf Y}$), whereas ${\textsf T}_y$ (${\textsf T}_x$) stabilises it, resulting in a 1D representation of $\mathbb{Z}_2$. 
	Thus, slightly above the BKT transition there are two sets of peaks giving two independent representations of the broken translation symmetry, corresponding to the two labellings above. As $T$ is increased, the $\pm{\textsf P}_{1,2}$ peaks broaden and disappear, see Fig.~\ref{fig:dimerU1}.c.
	
	For completeness, we note that  $({\textsf X},{\textsf Y})$ is sub-optimal in the ordered phase as it cannot deterministically map \emph{both} the site-parity \emph{and} the orientation into binary variables. 
	\footnote{Explicitly, it maps the two horizontal (vertical) columnar configurations to tuples of probabilities (for getting $+1$ for the given component) $\Lambda\cdot\mathcal{V}=(1\pm1,2)/2$ ($(2,1\pm1)/2$).} Thus these diagonal peaks are absent in Fig.~\ref{fig:dimerU1}.a.
	
	Yet, still a further important question can be asked. The above discussion was about lattice symmetries, we know though that in the effective field theory of the dimer model this discrete ${\textsf C}_4$ symmetry is in fact enlarged to a full $U(1)$. Can the ensemble of the RSMI-optimal filters provide a hint of this \emph{emergent symmetry}?
	
	Surprisingly, the answer is affirmative. This is despite the fact that the filters are trained on the lattice, and even despite the imposed discrete-valuedness of the coarse-grainings. The emergent continuous symmetry is represented on the space of kernels $\Lambda$ of the coarse-graining maps, and reflected in the ensemble.
	
	To show this we project the filters in the ensemble onto the orthogonal basis of the pristine filters $(E_x+E_y, E_x-E_y)$. As shown in Fig.~\ref{fig:dimerU1}.f, at high T the kernel distribution of the electrical fields returned in individual RSMI runs (all of which exhibit a, at first glance ``similar", diagonal pattern) is indeed invariant under continuous rotations!
	In fact, as plotted in Fig.~\ref{fig:dimerU1}.g, the RSMI is constant as a function of the angle $\theta$ of the \textit{rotated} electrical fields $R(\theta)(\Lambda_{\textsf S1}, \Lambda_{\textsf S2})\sim \left(\cos(\theta) E_x - \sin(\theta )E_y, \sin(\theta) E_x + \cos(\theta) E_y\right)$, thus demonstrating the connection between the (emergent) operator symmetry and the RSMI-degeneracy. For temperatures below the BKT transition, the filters do not overlap with the electrical fields (Fig.~\ref{fig:dimerU1}.d), while for those closely above (Fig.~\ref{fig:dimerU1}.e) the continuous $U(1)$ symmetry is not yet emerged, due to the effect of the plaquette operators in this finite system [contrast with Figs.~\ref{fig:dimerU1}.(b,c)].
	
	The above analysis of the ensemble can be improved by \emph{e.g.}~using more sophisticated methods to disentangle the mixtures more data-efficiently. Nevertheless, we conclude that even simple data analysis allows to extract the most important operators, and the symmetries (see also Ref.\onlinecite{bondesan2019learning}). 
	
	Finally, we remark that the filter ensemble may be a very useful and natural notion in disordered systems, where the RSMI approach also applies,\cite{optimalRSMI} and the filters may depend on the quenched disorder realization. We leave this intriguing possibility to future work.
	
	\subsection{Coarse-graining filters: type and number of components}\label{subsection:num_components}
	In the previous sections we demonstrated among others that the RSMI-optimal filters label ordered states, and how, more generally, the symmetries of the system, broken or emergent, even continuous, are manifested on the space of filters $\Lambda$. An important question is whether and how these results depend on the constraints imposed on the image of $\Lambda$, \emph{i.e.} the type of the coarse-grained variable $\mathcal{H}$. 

	There are, in fact, two questions: that of continuous- \emph{vs.}~discrete-valued $\mathcal{H}$, and that of the number of $\mathcal{H}$ components, which we address in this  section. We begin with the former, in particular explaining the surprising emergence of a continuous symmetry in ensemble of filters (as in the high $T$ dimers) even as $\mathcal{H}$ is discrete. 
	
		\begin{figure}[t]
		\centering
		\includegraphics[width=1\linewidth]{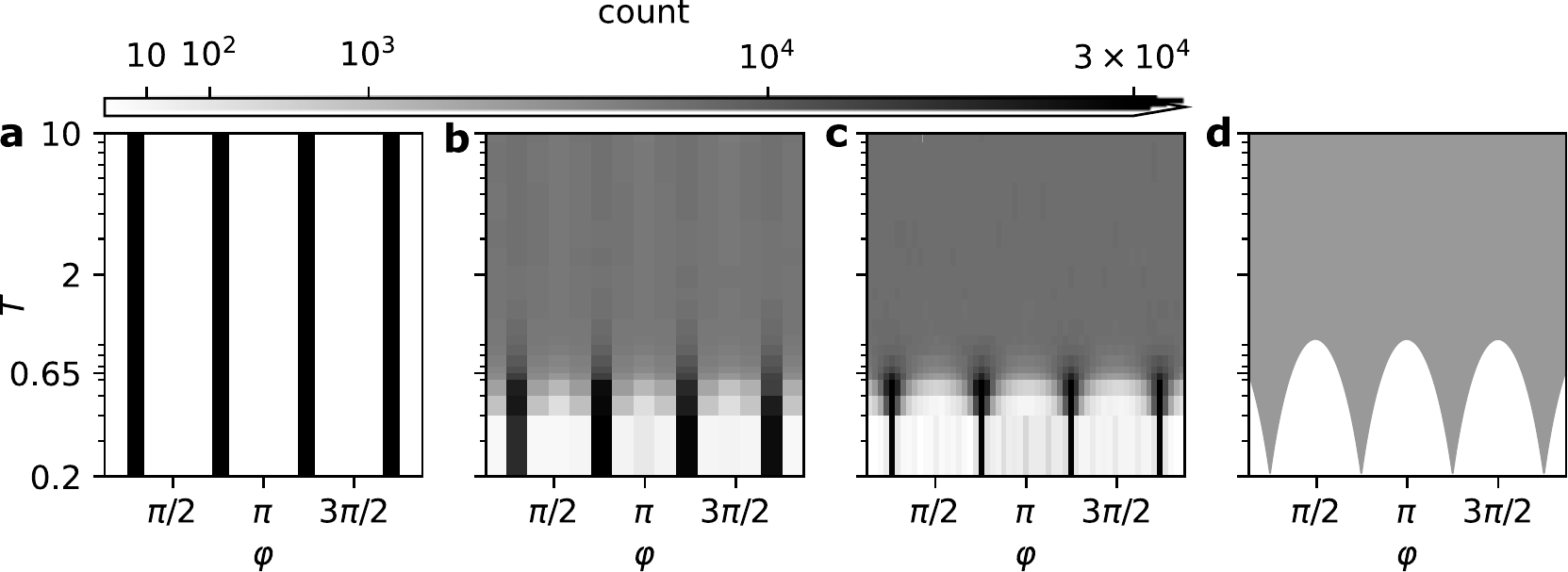}
		\caption{{\bf Coarse-graining smears the discrete lattice height field into a sine-Gordon field.} {\bf a} The histogram of values taken by the height field $\varphi(\mathbf{r}_i) \mod 2 \pi$ (defined on vertices $i$) for $T$ across the BKT transition; it is independent of temperature. {\bf b}, {\bf c} The histogram of values of the block-averaged field $\langle  \varphi \rangle_{\mathbf{r}_i \in \mathcal V}$ respectively for $4\times 4$ and $8\times 8$ blocks (corresponding to $L_\mathcal{V}=8, 16$ in dimer configurations) 
			{\bf d} The range of the continuous sine-Gordon field $\varphi(\mathbf{r})$.
		}
		\label{fig:dimerheightsupport}
	\end{figure}
		
	\subsubsection{Discrete and continuous variables}
	To this end, let us examine the role of coarse-graining (in a broad sense) in constructing a continuum low-energy effective field theory starting from a discrete lattice model, for which the dimer model is a good example.
	
	On any bipartite planar lattice, a dimer configuration $C(\mathbf{r}_i)$ can be mapped onto a unique height profile $\phi(\mathbf{r}_i)$.\cite{fradkin_2013}
	In 2D square lattice, the height at a given site is a four-state variable, invariant under a shift by $2\pi$, giving it a natural interpretation as the orientation of the dimer connected to the given site. Equivalently, the distribution of $\phi(\mathbf{r}_i)$ is uniform over four \emph{discrete} values (see Fig.~\ref{fig:dimerheightsupport}.a).

	The \emph{continuum} effective theory of the interacting dimer model is given by a sine-Gordon (SG) action:
	\begin{equation}\label{eq:sine-Gordon}
		S[\phi]=\int \mathrm{d}^2 \mathbf{r}\left[\frac{g(T)}{2}|\nabla \varphi(\mathbf{r})|^2 + \cos\left( 4 \varphi(\mathbf{r})\right) \right],
	\end{equation}
	 where the lattice height $\phi(\mathbf{r}_i)\in\mathbb{Z}_4$ is continued into a \emph{real-valued} $\phi(\mathbf{r})\in [0,2\pi)$. The first term in Eq.~\ref{eq:sine-Gordon} is the energy density associated to the electrical field $\mathbf{E}:=\nabla \varphi$, which accounts for the entropy of the dimers and dominates at high $T$. The cosine \emph{locking} potential orders $\varphi(\mathbf{r})$ into a flat profile with one of the four values $\{\pi/4, 3\pi/4, 5\pi/4, 7\pi/4\}$, corresponding to the four ground states of the lattice model. In other words, the BKT transition happens precisely at the temperature at which the range of $\varphi(\mathbf{r})$ fragments into four sectors, which eventually collapse into discrete delta-peaks at $T=0$, as shown in Fig.~\ref{fig:dimerheightsupport}.d.
	
	Numerically, the lattice and the SG pictures are bridged by mapping the discrete lattice heights onto \emph{block-averaged} values $\langle  \varphi \rangle_{\mathbf{r}_i \in \mathcal V}$ on the coarse grid. As seen in Figs.~\ref{fig:dimerheightsupport}.b and c, this simple procedure indeed smears out the discrete heights into a quasi-continuous variable of the same character as the SG field.
	The larger block size $L_{\mathcal V}$, the closer the coarse-grained variable to a continuous SG field.
	
	This observation is crucial in clarifying the role of discretisation in RSMI-NE, and sheds light on the scale $L_{\mathcal V}$.  Particularly, for the case of free dimers the RSMI filter computes the block average of the electric field:
	\begin{equation}
		\Lambda_\alpha \cdot \mathcal{V} = \sum_i \Lambda^i_\alpha \mathcal{V}_i = \sum_I \underbrace{\sum_{i\in I} \Lambda^i_\alpha \mathcal{V}_i}_{=:E^I_\alpha} =\langle E_\alpha \rangle_{\mathcal V},
	\end{equation}
	where $i$ iterates over bonds, and $I$ over disjoint smallest unit blocks (here this is $L_{\mathcal V}= 4$) on which the electric field can be defined; $\alpha$ denotes the electric field (or filter) component. On the other hand, in terms of the lattice height field, this means that (see Supplemental Material of Ref.\onlinecite{rsmine_letter}):
	\begin{equation}
		\Lambda_\alpha \cdot \mathcal{V} = \langle \partial_\alpha \varphi \rangle_{\mathcal V}=\partial_\alpha\langle \varphi\rangle_{\mathcal V},
	\end{equation}
	where $\partial_\alpha$ is the lattice gradient. In other words, the inner product computes the electrical field averages using the block-averaged heights (\emph{e.g.}~using $L_{\mathcal V}=8$, as we had in the discussion of the dimer model), which approximate the continuous ones of the SG theory. This is the reason why the emergent continuous $U(1)$ symmetry can manifest itself in the filter ensemble when the gradient part of the SG action dominates (in the finite system as $T\to \infty$), despite the discrete input data.

	Note that composing $\Lambda_\alpha \cdot \mathcal{V}$ with a binary discretisation map $\tau$ does not interfere with the $U(1)$ symmetry in the filter ensemble. We emphasize the same filters are obtained if no discretization is assumed. Discretisation becomes important, though, when the support of $\varphi$ starts to get locked into four values at low $T$. The information to be retained precisely comes from identifying these four peaks, and $\tau$ increases the efficiency of finding the filters bijectively mapping the symmetry-broken states by dramatically restricting the search space. In this sense it acts as an \emph{a posteriori} entropy cut-off/regularisation.
	
	We thus conclude that discrete-valuedness of coarse-graining maps does not interfere with or preclude continuous symmetries manifesting, and further serves as a regularisation, which is very useful when the number of configuration samples is limited.

	\subsubsection{Number of components}
	Next, we explain how the necessary number of hidden variable components can be \emph{discovered} systematically, providing information about the system. This particularly will justify using a single component variable in the Ising example, and two for the dimer model.
	
	The essence of the RSMI-NE approach is the efficient compression of the long-range information. Any compression method contains a trade-off (explicit or implicit) between the compression rate \emph{e.g.}~given by the total number of bits retained, and the preservation of relevant information. Ideally one should compress to preserve just sufficiently enough information relevant for the downstream task the compressed representation is used in -- but not more than that. In RSMI-NE the compression rate is effectively controlled by the number and the alphabet of allowed values of the coarse degrees of freedom $\mathcal{H}$ (for example $\{\pm1\}$, $\{\uparrow,\downarrow$\}, $\{\ket{0},\ket{1},\ket{2}\}$, etc.). The relevant long-range information, on the other hand, is a property of the physical system, which we do not have control over.
	
	\begin{figure}[]
		\centering
		\includegraphics[width=0.9\linewidth]{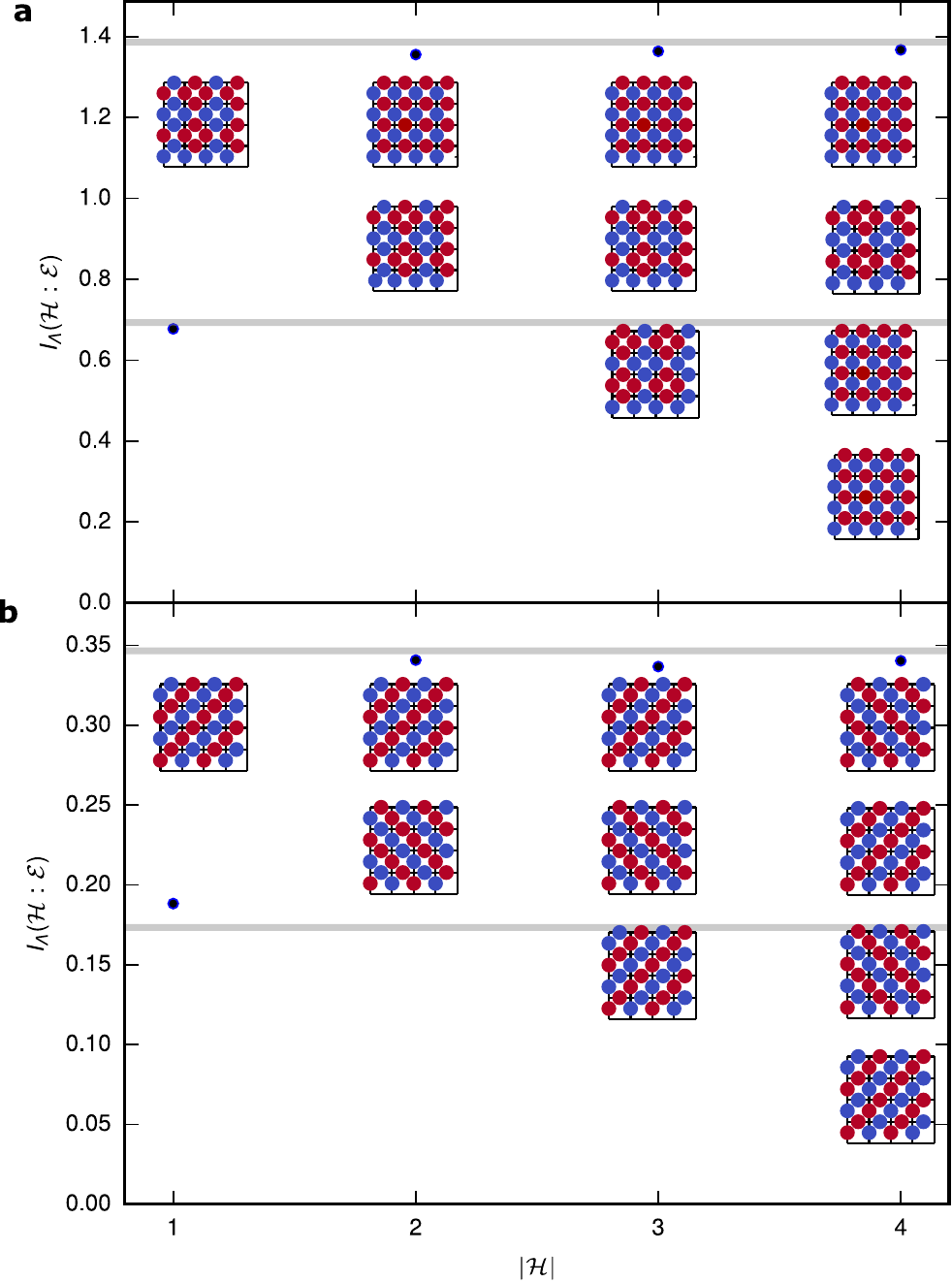}
		\caption{{\bf The maximal RSMI versus number of components of $\mathcal H$ in the dimer model.} {\bf a} For $T<T_{\ssm BKT}$, the only long-range information is about the type of ground-state the system crystallises into. Using two binary components for $\mathcal{H}$ suffices to encode this information and adding further components does not improve this result, as reflected by the value of $I_\Lambda(\mathcal{H}:\mathcal{E})$ attained. {\bf b} At $T\to \infty$, two electric field components recover the maximal long-range information $\approx \frac{1}{2}\log 2$ (for $L_{\mathcal{B}}=4$). This is not improved by additional components, which converge to filters linearly dependent on the first two. }
		\label{fig:dimercompscan}
	\end{figure}
	
	Since \emph{a priori} we do not know (though in practice we may often anticipate) what the amount of long-range information in the system is, a simple practical procedure to determine the optimal compression rate is to find the maximal number of compressed variables $|\mathcal{H}|$ above which the retained RSMI does not further increase significantly. Fig.~\ref{fig:dimercompscan} shows that this happens at $|\mathcal{H}|=2$ for the dimer model both in the high-$T$ limit, and the columnar ordered phase at low $T$. In both cases we use a fixed buffer width of $L_\mathcal{B}=4$. We find that using a single binary component leads to only half the RSMI value attained with $|\mathcal{H}|=2$, and is obtained by an optimal filter equal to one of the components of the two-component rule. For $|\mathcal{H}|>2$, the RSMI saturates into either $\frac{1}{2}\log2$ at $T\to \infty$ or $\log 4$ at $T<T_{\ssm BKT}$. Thus we verify that at most two of the components are linearly independent, and additional filters do not extract distinct information from $\mathcal{V}$. In the 2D Ising model example, in contrast, only a single filter suffices (and it corresponds to magnetization).
 	The physical intuition behind the above procedure is clear: it finds the number of relevant operators whose correlations explain the total information shared between distant parts of the system.
	
	Finally, we note that while the variational \emph{ansatz} used for the coarse-graining was a shallow network dotted into the configurations (before the non-linear Gumbel-softmax step), more general or deeper network architectures can be considered. RSMI maximization can be performed over any class of variational functions. For specific systems/inputs a more expressive \emph{ansatz} could possibly recover larger RSMI. In more abstract terms, since the coarse-graining rules are related to the RG-relevant operators,\cite{rsmine_letter} such choices would be able to extract operators which cannot be written as linear functions of the local degrees of freedom, should these be important. Though for such complex multi-layered architectures the patterns of the weights themselves may not be directly interpretable, the extracted $\Lambda$ filters can still be used as operators, as we have done in computing the correlation functions in Ref.~\onlinecite{rsmine_letter}. We again emphasize that the physical interpretability is not a consequence of the architectural choices, but of the physical nature of the RSMI quantity.\cite{Gordon2020}

	\section{Possible extension to non-equilibrium: a model with chipping and aggregation}\label{section:chipping}
	
	The RSMI-NE algorithm we described does not in any way use or rely on the existence of a Hamiltonian generating the probability distribution. It can thus be directly applied to general  non-equilibrium distributions, though the formal understanding of the optimal filters in this situation, analogous to results of Refs.~\onlinecite{Gordon2020,rsmine_letter}, is currently missing.
	This is an exciting research direction, whose development we leave to future work. Here, however, we provide a short validation of the idea. 
	
	To this end we consider the non-equilibrium example of the 1D chipping and aggregation model of Ref.~\onlinecite{PhysRevE.63.036114}. Its stochastic dynamics is defined by the update rules given below. At any time increment $\Delta t$, masses $m_i$ on site $i$ in the chain are modified according to the following moves:
	\begin{widetext}
	\begin{align}
		\text{chipping, at rate } p=\Delta t:\hspace{0.5 cm} &
		\begin{cases}
			\text{if }m_i>0&: \, m_i\mapsto m_i-1, \hspace{0.5 cm} m_{i\pm 1}\mapsto m_{i\pm1}+1,\\
			\text{if }m_i=0&:\, \text{do nothing},
		\end{cases}\\
		\text{aggregation, at rate }p=w\Delta t:\hspace{0.5 cm}  &m_i\mapsto 0, \hspace{0.5 cm} m_{i\pm1}\mapsto m_{i\pm1}+m_i.
	\end{align}

	\begin{figure*}[]
		\centering
		\includegraphics[width=\linewidth]{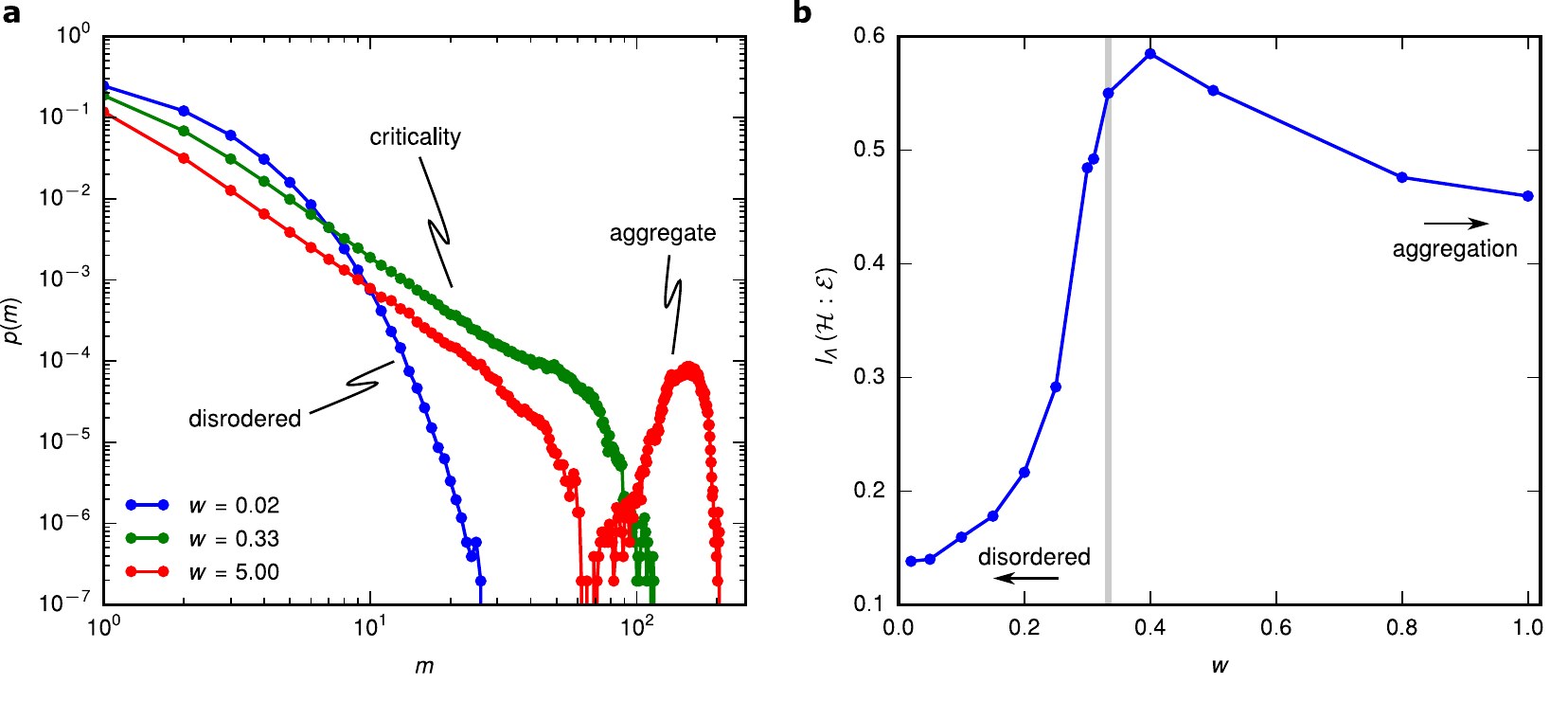}
		\caption{{\bf Phase transition in the chipping and aggregation model.} {\bf a} Observing the decay profile of the marginal mass distribution for a given site at different values of $w$ is one way of qualitatively assessing the different phases of the chipping model. {\bf b} By providing merely the real-space sample configurations, the peaking maximal value of RSMI signals the non-equilibrium critical point. Furthermore, the aggregated and low-$w$ phases can be distinguished by the distinct saturated values.}
		\label{fig:rsmichipping}
	\end{figure*}
	\end{widetext}
	We consider the case with mass density $\rho=1$, so that $\sum_i m_i=L$, where $L$ is the length of the chain for which we impose periodic boundary conditions.
	
	This system undergoes a non-equilibrium transition between phases, whose qualitative difference is reflected in the marginal probability distribution $p(m)$ of masses at each site, as shown in Fig.~\ref{fig:rsmichipping}.a. For low aggregation $w$, the distribution $p(m)$ is an exponentially decaying function with increasing mass per site. At high $w$, a macroscopic proportion of the masses aggregate (or condense) to a single site, corresponding to a delta-function peak around $m\approx L$, with an algebraically decaying part for the remainder of the masses. Criticality of the intermediate $w$ region is reflected by a purely algebraically decaying distribution $p(m)$.
	
	The shape of the marginal distribution $p(m)$ clearly demonstrates the different phases and the transition between them. Instead of, however, investigating such specific quantities, which requires at least an intuitive understanding of the physics of the system, one can generically identify the non-equilibrium critical point using RSMI-NE directly applied to the full real-space configurations. As shown in Fig.~\ref{fig:rsmichipping}.b the dependence of the optimal-RSMI at $w$ readily points to a transition between the phases with different spatial correlation characteristics. We performed RSMI optimisation on the model with $L=256$, for a range of values of $w$, and used a buffer of width $L_\mathcal{B}=8$. While at all values of $w$ we found that the optimal coarse-graining rule averages the mass for the region $\mathcal{V}$, the value of the maximal information saturates into different values in the low-$w$ and the aggregated phases. Moreover, the transition point is marked by a slight peak in the RSMI.
	
	These results demonstrate the in principle applicability of the RSMI maximisation to non-equilibrium problems. While this example reaches the steady-state distribution very quickly, in a more general scenario involving far-from-equilibrium systems the formalism can be extended to screen out short-time correlations by introducing a buffer in the temporal direction. This line of research merits a full development, which we leave for future study.

	\section{Conclusions and outlook}\label{conclusion}
	
	In this work, accompanying and extending Ref.~\onlinecite{rsmine_letter}, we demonstrate how recent rigorous results in ML-based estimation of information theoretic quantities\cite{belghazi2018mine,poole2019variational} can be combined with other algorithmic ingredients\cite{gumbel1954,NIPS2014_5449} to yield a formally interpretable and numerically efficient algorithm extracting information about long-range properties of statistical systems from its raw configurational samples. We dub this unsupervised algorithm RSMI-NE, or the Real-space Mutual Information Neural Estimator, from the key physical quantity of interest,\cite{Koch-Janusz2018, optimalRSMI, Gordon2020} and provide a detailed introduction to the method and background concepts, as well as an examination of its properties.
	
	The optimal real-space coarse-graining filters extracted depend on the parameters of the physical system, and in fact are lattice representations of its most relevant operators in the sense of renormalization group (RG).\cite{rsmine_letter} The filters, along with the RSMI value, and its dependence of systems' parameters and length-scales, provide a comprehensive physical picture. Particularly, the position and nature of the critical points, decay and structure of the spatial correlations, together with the type of order is revealed by these quantities, which we demonstrate explicitly in equilibrium examples.
	Though we focused on statistical models on regular lattices, the algorithm can be applied to continuum or graph models as well.
	
	We further introduce the notion and examine the properties of the \emph{ensemble of the optimal coarse-graining filters}. We show that it contains important physical information, most notably about the symmetries, including \emph{emergent} ones. The ensemble can be the object of statistical analysis itself, allowing to extract the relevant operators from partially complete data, or from restricted parameter regimes.
	
	We also examined and validated the possibility of extending the applicability of the algorithm to non-equilibrium systems on the example the chipping model with aggregation,\cite{PhysRevE.63.036114} for which the presence and position of a non-equilibrium phase transition was detected. 
	
	Motivated by the above example, the full extension of the framework to the case of non-equilibrium systems is among the most promising and exciting future research directions. This requires investigating the RSMI approach with spatio-temporal coarse-graining rules, and the information shared between temporally ordered states\cite{PhysRevE.79.041925} and extending the theoretical results of Ref.~\onlinecite{Gordon2020}. The existing algorithm does, however, already allow the investigation of spatial correlations in complex real-world systems, such as \emph{e.g.}~meteorological precipitation data exhibiting critical points, possibly related to self-organized-criticality.\cite{Peters2006}
	
	Previous formal results\cite{optimalRSMI} and the numerical possibility of individually optimising the coarse-graining rules for each block invite the application of RSMI-NE to disordered systems. In particular, the ensemble of coarse-graining rules would be inherited naturally from the disorder distribution and the further statistical analysis may identify certain equivalence classes within this ensemble. This would be especially helpful in identifying the relevant DOFs in these challenging systems.
	
	Finally, we emphasize that the RSMI-NE provides an important step towards the goal of automating certain aspects of theory building. The constructed outputs (the coarse-graining transformations) are effectively black-box algorithmic objects, which, however, can be assigned the formal identity of order parameters or scaling operators of the physical theory. They can explicitly be used as such to compute correlations functions or scaling exponents, as shown in the companion work Ref.~\onlinecite{rsmine_letter}. The above results clearly invite further work in this direction.

	
	\bigskip
	\noindent
	{\bf Code availability} Source code and an example Python notebook for the RSMI-NE package are available online at \url{https://github.com/RSMI-NE/RSMI-NE}.
	
	\smallskip
	\noindent
	{\bf Acknowledgements} M.K.-J. is grateful to F. Alet for his comments on the physics of the interacting dimer model. We acknowledge insights into the physics of the chipping and aggregation model coming from an earlier study of this system together with R. Thomale. D.E.G., S.D.H., and M.K.-J. gratefully acknowledge financial support from the Swiss National Science Foundation and the NCCR QSIT, and the European Research Council under the Grant Agreement No.~771503 (TopMechMat), as well as from European Union's Horizon 2020 programme under Marie Sklodowska-Curie Grant Agreement No.~896004 (COMPLEX ML). Z.R. acknowledges support from ISF grant 2250/19. Some of the computations were performed using the Leonhard cluster at ETH Zurich. This work was supported by a grant from the Swiss National Supercomputing Centre (CSCS) under project ID eth5b.

	\ifpreprint%
		\clearpage
	
		\appendix
		\begin{widetext}
			\begin{center}
				{\normalsize \bf Supplemental Material: Symmetries and phase diagrams with real-space mutual information neural estimation}
			\end{center}
			
			\input{supp_text_pre}
		\end{widetext}
	
	\else%

	\fi%

	\bibliographystyle{naturemag}
	\bibliography{ref}
	

\end{document}

%% file: supp_text_pre.tex
	\section{Some properties of mutual information}
	The Shannon \textit{mutual information}\index{mutual information} (MI) of two random variables $\mathcal{X}$ and $\mathcal{Y}$ quantifies the decrease in entropy of one of the random variables when the other one is observed. Equivalently, it is the amount of knowledge we gain about one of them, when observing the other. Formally it is defined as a difference of entropies:
	\begin{equation}
		I(\mathcal{X}:\mathcal{Y}):=H(\mathcal{X})-H(\mathcal{X}|\mathcal{Y}) = H(\mathcal{X}) + H(\mathcal{Y}) - H(\mathcal{X},\mathcal{Y}).
	\end{equation}
	The above expression indicates that the real-space mutual information (RSMI) can take values at most on the order of a few units of information when coarse-graining small blocks $\mathcal{V}$ that contain $N_\mathcal{V}$ individual degrees of freedom with a discrete alphabet of $n$ symbols. Indeed:
	\begin{equation}\label{RSMI_is_small}
		I_\Lambda(\mathcal{H}:\mathcal{E})\leq I(\mathcal{V}:\mathcal{E})=H(\mathcal{V})-H(\mathcal{V}|\mathcal{E})\leq H(\mathcal{V})\leq N_\mathcal{V}\log n,
	\end{equation}
	where, since $\Lambda$ compresses $\mathcal{V}$ into $\mathcal{H}$, the first inequality follows from the data-processing inequality and the second inequality follows from the positive semi-definiteness of Shannon entropies. This property ensures the applicability of the MI estimation methods by maximising variational lower-bounds, which can suffer from a bias-variance trade-off in the opposite regime,~\emph{i.e.}~when the MI is large.
	
	By expanding the entropies, we recover an alternative expression for $I(\mathcal{X}:\mathcal{Y})$:
	\begin{equation}
		I(\mathcal{X}:\mathcal{Y})=:D_{\rm KL}\left[p(x,y)||p(x)p(y)\right],
	\end{equation}
	where $D_{\rm KL}$ is the Kullback-Leibler (KL) divergence\footnote{The Kullback-Leibler divergence is a measure of distance (but formally not a metric) between the two probability distributions in its argument.}. The Gibbs' inequality $D_{KL}(p||q)\geq 0$ predicates on a useful interpretation of MI: given a pair of random variables $(\mathcal{X},\mathcal{Y})$ jointly distributed according to $p(x,y)$, $I(\mathcal{X}:\mathcal{Y})$ measures the information lost when encoding $(\mathcal{X},\mathcal{Y})$ as a pair of independent random variables while they may not be so. This loss is $0$ if and only if $\mathcal{X}$ and $\mathcal{Y}$ are actually independent of each other, \emph{i.e}~when $p(x,y)=p(x)p(y)$.
	
	\subsection{Log concave bound and TUBA}
	The issue of estimating the log partition function appearing in the UBA bound is circumvented by taking advantage of the concavity of the logarithm. As discussed in the main text, this lead to the so-called tractable unnormalised BA (TUBA) lower-bound, first derived by Poole \textit{et al.~}in Ref.~\onlinecite{poole2019variational}. Here we briefly expand the further details of the derivation.
	
	Since $\log$ is a strictly concave function, by the mean value theorem we have:
	\begin{align*}
		\log x - \log 1 &\leq (x-1) \frac{{\rm d}}{{ \rm d}x}\log x\big{|}_{x=1}\\
		\log x&\leq x - 1, \hspace{0.5cm} x>0.
	\end{align*}
	This implies that:
	\begin{equation}\label{log_concave_bound}
		\log Z = \log \frac{Z}{a} + \log a \leq \log a + \frac{Z}{a} - 1, \hspace{0.5cm} Z,a>0.
	\end{equation}
	Substituting the RHS of this inequality in the UBA lower-bound in Eq.~(\ref{UBA}), we obtain the TUBA lower-bound:
	\begin{align}
		I_{\rm UBA}(\mathcal{X}:\mathcal{Y}) &\geq \mathbb{E}_{p(x,y)}[f(x,y)] -  \mathbb{E}_{p(y)}\left[\log a(y) + \frac{Z(y)}{a(y)} - 1\right]\nonumber\\
		&\geq \mathbb{E}_{p(x,y)}[f(x,y)] - \mathbb{E}_{p(x)p(y)}\left[\frac{e^{f(x,y)}}{a(y)}\right] - \mathbb{E}_{p(y)}[\log a(y)] -1 =: I_{\rm TUBA}(\mathcal{X}:\mathcal{Y}).
	\end{align}
	
	\subsection{Extremum and the tightness of the NWJ bound}\label{NWJ_extremum}
	The extremum $f^*(x,y)$ of $I_{\rm NWJ}(\mathcal{X}:\mathcal{Y})=I_{\rm NWJ}(\mathcal{X}:\mathcal{Y})[f(x,y)]$ is found by setting its functional derivative to $0$:
	\begin{align}
		&0\stackrel{!}{=}\frac{\delta}{\delta f(x,y)}\left(\sum_{x,y}p(x,y)f(x,y)-e^{-1}\sum_{x,y}p(x)p(y)e^{f(x,y)}\right)\Bigg{|}_{f(x,y)=f^*(x,y)}\\
		\implies\text{either}& \hspace{0.5cm} p(x,y) - e^{-1}e^{f^*(x,y)}p(x)p(y) = 0 \hspace{0.5cm} \forall x,y, \hspace{0.5cm} \text{(maximum),} \nonumber\\
		\text{or}& \hspace{0.5cm} f^*(x,y)=1\hspace{0.5cm} \forall x,y, \hspace{0.5cm} \text{(minimum),} 
	\end{align}
	which implies that the optimal \textit{ansatz} is:
	\begin{equation}\label{opt_NWJ}
		\hspace{0.5cm} f^*(x,y)=1+\log \frac{p(x,y)}{p(x)p(y) }.
	\end{equation}
	Substituting this in Eq.~(\ref{NWJ_bound}), we thus find that the NWJ lower-bound is tight when $f=f^*$, \emph{i.e.}, $I_{\rm NWJ}(\mathcal{X}:\mathcal{Y})[f=f^*]=I(\mathcal{X}:\mathcal{Y})$: in this case we have:
	\begin{align*}
		Z[f^*](y)=\sum_x p(x) e^{f^*}=e\sum_x p(x) \frac{p(x,y)}{p(x)p(y)}
		=e,
	\end{align*}
	and since $a(y)=e$, the log concave inequality~\ref{log_concave_bound} becomes tight.
	
	\subsection{Upper bound of InfoNCE}\label{NCE_upper}
	Even though it is manifest in the derivation that InfoNCE is a lower-bound to the mutual information, it need not be a tight bound. To see this we can express InfoNCE as:
	\begin{equation*}
		I_{\rm NCE}(\mathcal{X}:\mathcal{Y})[f]=\mathbb{E}_{\prod_{k=1}^K p(x_k, y_k)} \left[\frac{1}{K}\sum_{j=1}^K\left( f(x_j, y_j) -\log \sum_{i=1}^K e^{f(x_i, y_j)} \right) \right]+\log K.
	\end{equation*}
	Since $\sum_{j=1}^K e^{f(x_i,y_j)}>e^{f(x_j,y_j)}$, and since the logarithm is a monotonically increasing function for positive arguments, we have:
	\begin{align}
		I_{\rm NCE}(\mathcal{X}:\mathcal{Y})[f]<&\mathbb{E}_{\prod_{k=1}^K p(x_k, y_k)} \left[\frac{1}{K}\sum_{j=1}^K\left( f(x_j, y_j)  -\log  e^{f(x_j, y_j)} \right) \right]+ \log K\nonumber\\
		&=\log K, \hspace{0.5cm} \forall g.
	\end{align}
	
	Thus, InfoNCE is bounded from above by $\log K$. In other words, the InfoNCE bound is not tight if the number of replicas is not sufficiently large compared to the value of the mutual information or, more precisely, when:
	\begin{equation}\label{InfoNCE_biasvariance}
		e^{I(\mathcal{X}:\mathcal{Y})}>K.
	\end{equation}
	Note that, as we mentioned above, in the regime the RSMI-NE algorithm is working this is not a concern, since the real-space mutual information is at most a few bits. 
	
	\subsubsection{Maximal value of InfoNCE (further properties)}\label{NCE_practical}
	The expectation value in InfoNCE is taken over multiple samples of the $K$-replica random variable $(x_{1:K}, y_{1:K})$. A single $2K$-dimensional replica sample can be considered as a minibatch of $K$ samples, each drawn from $p(x,y)$. Therefore, for a total of $nK$ samples drawn from $p(x,y)$, we compute the InfoNCE bound by practically forming an $n$-sample MC estimate for the expectation value of the expression:
	\begin{equation}
		\sum_{j=1}^K \log \frac{e^{f(x_j,y_j)}}{\sum_{i=1}^{K}e^{f(x_i,y_j)}}.
	\end{equation}
	Here, the argument of the logarithm is known as the softmax function\index{softmax function}:
	\begin{equation}\label{eq:si_softmax}
		{\rm softmax}_j (\mathbf{v}):=\frac{e^{v_j}}{\sum_i e^{v_i}},
	\end{equation}
	and defining the \textit{prediction} $Q:=\prod_{j=1}^K{\rm softmax}_j\left(f(x_{1:K},y_j)\right)^{1/K}$, we arrive at:
	\begin{equation}\label{NCE_expectant}
		\sum_{j=1}^K \log \frac{e^{f(x_j,y_j)}}{\sum_{i=1}^{K}e^{f(x_i,y_j)}}=\log \prod_{j=1}^K{\rm softmax}_j\left(f(x_{1:K},y_j)\right)=K\log Q(x_{1:K},y_{1:K}).
	\end{equation}
	Note that the arguments $(x_i,y_j)$ of $g$ with $i\neq j$ correspond to samples that belong to separate minibatches, whereas $(x_j,y_j)$ denote joint samples. Consequently, by inspecting Eq.~(\ref{NCE_expectant}), we see that maximising the InfoNCE bound requires a $f(x_i,y_j)$ that can discriminate jointly and independently drawn samples (\textit{cf.} TUBA and NWJ bounds). 
	
	More precisely, if $p(x,y)\neq p(x)p(y)$, the product of softmax functions selects $f(x_j,y_j)$ which takes the largest relative value compared to other $f(x_{k\neq j},y_j)$ for all $j$, in which case the logarithm goes to $0$ and $I_{\rm NCE}(\mathcal{X}:\mathcal{Y})$ gets closer to its maximal value (either $I(\mathcal{X}:\mathcal{Y})$ or $\log K$). On the contrary, if $p(x,y)= p(x)p(y)$, then it is impossible to distinguish independent and joint samples, and $g$ takes similar values for all arguments. In this case $Q(x_{1:K},y_{1:K})$ becomes roughly uniform, i.e., $=1/K$ and $I_{\rm NCE}(\mathcal{X}:\mathcal{Y})$ vanishes, as it should. In other words, InfoNCE is an effective binary classifier of bivariate probability distributions, distinguishing whether they are a product of marginals or not.
	
	In fact, up to an additive constant, the InfoNCE bound is simply equal to the categorical cross-entropy\footnote{The cross-entropy function is commonly used in ML as an objective function.} $H[P,Q]$, for correctly distinguishing a joint sample from all $K-1$ independent samples. That is:
	
	\begin{equation}\label{CCE}
		H[P, Q]:=-\mathbb{E}_{P(x_{1:K},y_{1:K})}\left[\log Q(x_{1:K},y_{1:K})\right] = -I_{\rm NCE}(\mathcal{X}:\mathcal{Y})+\log K,
	\end{equation}
	with $Q(x_{1:K},y_{1:K})$ being the \textit{prediction} of InfoNCE and $P(x_{1:K},y_{1:K})=P(x_{1:K},y_{1:K}):=\prod_{i=1}^K p(x_i, y_i)$ is the product distribution of the $2K$-dimensional replica sample $(x_1,\cdots,x_K,y_1,\cdots y_K)$.
	

\section{Further details of the RSMI-NE algorithm}
\subsection{Gumbel-softmax reparametrisation of discrete random variables}

\subsubsection{Gumbel-max reparametrisation}
Let $\{\pi_i\}_{i=1}^N$ be an $N$-state categorical distribution, where $\pi_i$ denotes the probability of drawing a sample in $i$'th state. Furthermore let us define the Gumbel (or the double-exponential) distribution centred at $\mu$:
\begin{equation}
	p_\mu(z):=\exp(-z+\mu)\exp\left(-\exp(-z+\mu)\right).
\end{equation}
The corresponding cumulative distribution function (CDF) is given by:
\begin{equation}
	P_\mu(z):=\int_{-\infty}^z {\rm d}z' p_\mu(z)=\exp\left(-\exp(-z+\mu)\right),
\end{equation}
which is the probability of drawing a random variable $g\sim p_\mu(g)$ that is smaller than $z$. It follows that a standard Gumbel random variable $g$ can be obtained by transforming a standard uniform random variable $u$ by 
\begin{equation*}
	g=-\log(-\log u).
\end{equation*}

Given the definitions above, we will now prove the following Lemma due to Refs.~\onlinecite{gumbel1954,NIPS2014_5449}:
\begin{lemma}
	Let $\{g_i\}_{i=1}^N$ be a set of independent and identically distributed (i.i.d.) samples drawn from $p_0(g_i)$. Then
	\begin{equation}
		k^*=\argmax_{k\in\{1:N\}}\left\{g_i + \log \pi_i\right\}_{i=1}^N \hspace{0.5 cm} 
	\end{equation}
	is a random variable that is drawn from the categorical distribution $\{\pi_i\}_{i=1}^N$.
\end{lemma}
\begin{proof}
	Let $x_i := \log \pi_i$. Suppose that we draw a set of $N$ i.i.d. samples $\{z_i \sim P_{x_i}(z_i)\}_{i=1}^N$. Let $k^*$ be defined such that $z_{k^*}>z_{i\neq k^*}$. Given that the $k$'th sample takes the value $z_k$, the probability for $z_k$ being the greatest among all $N$ samples is then given by (using the definition of the Gumbel CDF):
	\begin{equation}
		p(k = k^* | z_k) = \prod_{i\neq k} P_{x_i}(z_k).
	\end{equation}
	Because of the Bayes' rule we have:
	\begin{equation*}
		p(k=k^*, z_k)=p(k=k^*|z_k)p(z_k),
	\end{equation*}
	so the probability of the $k$'th sample taking the greatest value $\forall z_k$ can be obtained by marginalising this over $z_k$
		\begin{align}
			p(k = k^*) &= \int_{-\infty}^\infty {\rm d}z_k p_{x_k}(z_k)\prod_{i\neq k}P_{x_i}(z_k)\nonumber\\
			&= \int_{-\infty}^\infty  {\rm d}z_k \exp(-z_k+x_k)\exp\left(\exp(-z_k+x_k)\right)\prod_{i\neq k}\exp\left(-\exp(-z_k+x_i)\right)\nonumber\\
			&= \exp(x_k) \int_{-\infty}^\infty  {\rm d}z_k \exp\left[ -z_k  - \exp(-z_k)\sum_{i=1}^N \exp(x_i)\right]\nonumber\\
			&= \exp(x_k) \int_{-\infty}^\infty \frac{{\rm d}}{{\rm d}z_k} \frac{\exp\left(-\exp\left(-z_k\sum_{i=1}^N \exp(x_i)\right)\right)}{\sum_{i=1}^N \exp(x_i)}\nonumber\\
			&=\frac{\exp(x_k)}{\sum_{i=1}^N \exp(x_i)}=\pi_k.
		\end{align}
	In other words, we have found that the index of the random sample with the largest value from $\{z_i \sim P_{x_i}(z_i)\}_{i=1}^N$ is distributed according to the categorical distribution defined by $\{\pi_i\}_{i=1}^N$. By definition, we have:
	\begin{equation}
		k^*=\argmax_{k\in\{1:N\}}\{z_i \sim P_{x_i}(z_i)\}_{i=1}^N \sim \{\pi_i\}_{i=1}^N.
	\end{equation}
	We can reparametrise the random variable $z \sim P_{x_i}(z)$ in terms of the parameter $x_i$ of the distribution and the standard Gumbel random variable $g\sim P_0(g)$ as simply $z=g+x_i$. Recalling that $x_i=\log \pi_i$, we finally arrive at the desired result
	\begin{equation}
		k^*=\argmax_{k\in\{1:N\}}\{g_i + \log \pi_i\}_{i=1}^N \sim \{\pi_i\}_{i=1}^N.
	\end{equation}
\end{proof}

\subsubsection{Gumbel-softmax reparametrisation}
Still, the argmax function is not differentiable with respect to $\pi_i$. The idea \cite{jang2016categorical, maddison2016concrete} allowing to circumvent this issue is to smear-out the argmax, replacing it by the softmax function. Given $\{g_i\}_{i=1}^N$ we define a vector-valued random variable utilizing the softmax function Eq.~(\ref{eq:si_softmax}), whose j-th component takes the form:
\begin{equation}
	{\rm softmax}_{j,\epsilon}\left(\{g_i + \log \pi_i\}_{i=1}^N\right)=\frac{\exp\left[(\log \pi_j + g_j)/\epsilon\right]}{\sum_{i=1}^N\exp\left[(\log \pi_i + g_i)/\epsilon\right]},
\end{equation}
where $\epsilon$ is the smearing parameter. For $\epsilon \rightarrow 0$ the softmax becomes the argmax function, mapping the argument vector $y=\{g_i + \log \pi_i\}_{i=1}^N$ into a $N$-component one-hot vector (one-hot encoding maps each of $N$ possible states $i$ of a discrete variable into a $N$-dimensional vector, with $1$ on $i$-th position, and zeros elsewhere) with some $k^*$-th entry taking the value $1$, thereby marking $y_{k^*}=\max y$. 
The resulting random variable is called a Gumbel-softmax random variable; it is only approximately (or pseudo-) discrete, for small enough $\epsilon$ (do not confuse with a discrete random variable defined by taking the maximum component of the softmax function). 
In practice though, the error coming from using a finite $\epsilon$ can be made comparable to machine precision. In the next subsection we explain the full training procedure in more detail. The samples from the Gumbel-softmax approximation of a certain categorical distribution $\{\pi_k\}_{k=1}^N$ are approximately one-hot vectors for small $\epsilon$. More concretely, for $\epsilon\approx0$, a sample vector $h\sim{\rm softmax}_\epsilon(\{g_i + \log \pi_i\}_{i=1}^N)$ has a single component very close to $1$ and all other components take very small values. Correspondingly, the expectation value of the samples $h$ generate almost exactly the set of frequencies $\{\pi_k\}$ of the categorical distribution, see \emph{e.g.} $\epsilon=0.01$ in  Fig.~\ref{gumbelsoftmax_relaxation}.  With increasing $\epsilon$, the samples deviate further from the one-hot form as multiple of their components start to take finite values (see $\epsilon=0.6,10$ in \emph{e.g.} Fig.~\ref{gumbelsoftmax_relaxation}) and the expectation value for the samples deviates from the original distribution $\{\pi_k\}$ as it becomes uniform over all components in the limit $\epsilon\to \infty$.

\subsubsection{Annealing of the Gumbel-softmax}\label{GS_annealing}
There is a trade-off between small $\epsilon$ which leads to very noisy gradient estimates, and large $\epsilon$ at which the gradients have low variance but the samples are far from being discrete. To reconcile this, we start the training at a high value of $\epsilon$ and anneal it towards a small positive value during training and stiffen the pseudo-discrete variable into an increasingly better approximate of a discrete one. This is illustrated in Fig.~\ref{gumbelsoftmax_relaxation}. A discrete random variable (drawn from a categorical distribution), which takes a value corresponding to one of $N=5$ categories can be represented by a one-hot vector $h=\{h_k\}_{k=1}^N$.

\begin{figure*}[]
	\centering
	\includegraphics[width=0.85\linewidth]{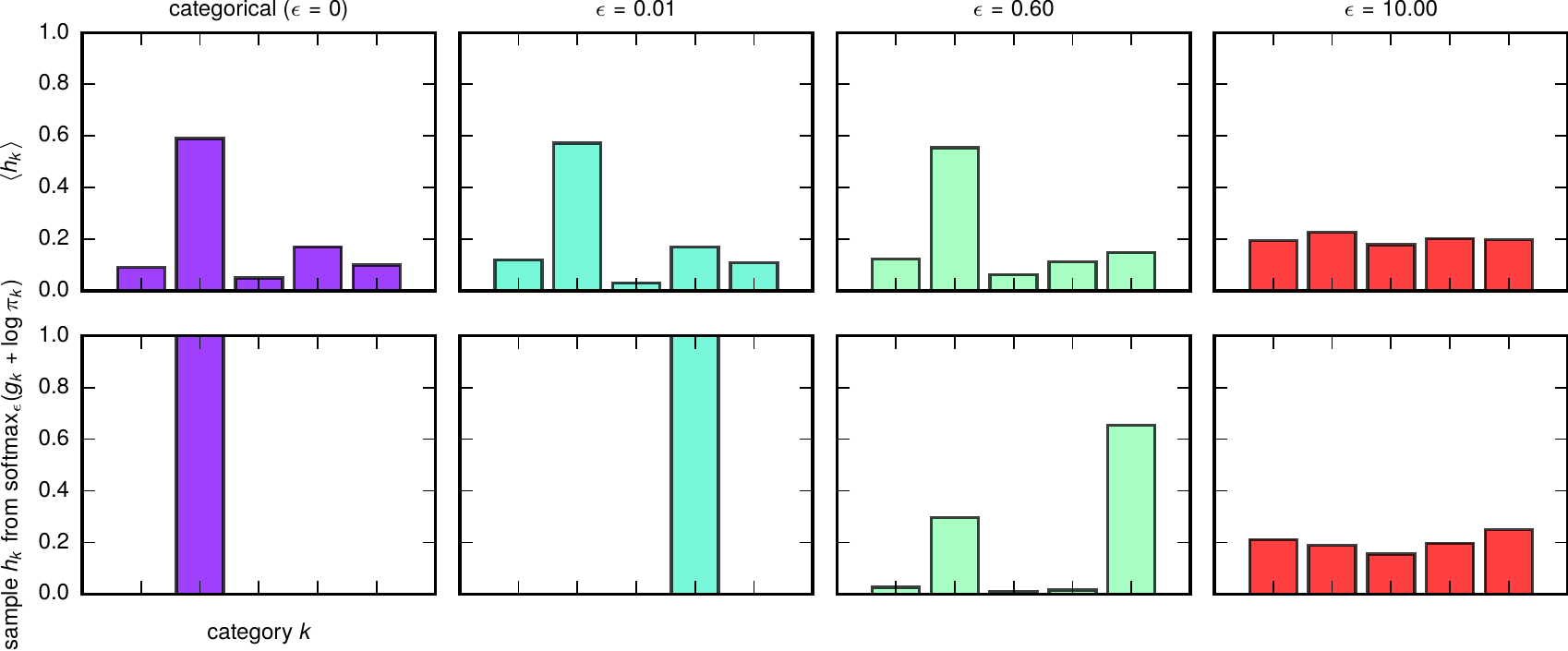}
	\caption{{\bf The Gumbel-softmax distribution for different values of the relaxation parameter.} Upper panel: Expectation value of the sample vectors from the Gumbel-softmax (GS) distribution with 5 categories, for different values of the parameter $\epsilon$. Lower-panel: single samples drawn from the GS distribution with the $\epsilon$ above. They are approximately one-hot when $\epsilon$ is small, see the text. For comparison, the categorical case where the sample is exactly a one-hot vector is shown in the left-most column. Figure adapted from Ref.~\onlinecite{jang2016categorical}.}
	\label{gumbelsoftmax_relaxation}
\end{figure*}

We have empirically observed that the annealing allows quicker convergence during training compared to a small fixed relaxation parameter, if the annealing schedule is chosen appropriately. More concretely, letting $t$ be the current step of the training, we have opted for an exponential scheduling of the form:
\begin{equation}
	\epsilon_t = \max(\epsilon_{\rm min}, \epsilon_{\rm max}\exp(-r t)).
\end{equation}
The parameter $r$ and the initial and the minimum values of $\epsilon$ are determined by experimentation.

In Tab.~\ref{CG_arch} we give the details of the architecture for the coarse-graining network used for the 2D Ising and 2D dimer models. In both cases we stack a single layer $\Lambda$-filter (generally with multiple kernels, corresponding to different components of $\mathcal{H}$) and the Gumbel-softmax reparametrisation layer to embed the components of $\mathcal H$ into (pseudo-) binary variables. While we determine the relaxation parameter $r$ by experimentation and fix it for both models, we tune the initial value of the Gumbel-softmax temperature $\epsilon_{\rm max}$ according to the total number of iterations during training. 

\begin{table}[H]
	\centering
	\caption{\textbf{Architecture details of the coarse-graining module of RSMI-NE for 2D Ising (anti-)ferromagnet and 2D interacting dimer model on a square lattice.}}
	\begin{tabular}{|c||c|c|}
		\hline
		model & 2D Ising  & 2D dimer \\
		\hline
		\hline
		$L_{\mathcal{V}}$ & $4$(F),$5$(AF)  & 8  \\
		\hline
		$L_{\mathcal{B}}$ & $\{0,\cdots,8\}$ & $\{2,4,6,8\}$ \\
		\hline
		$L_{\mathcal{E}}$ & 10 & 4  \\
		\hline
		number of components of $\mathcal{H}$ & 1 & 2   \\
		\hline
		embedding of $\mathcal H$ & binary &  binary \\
		\hline
		$(\epsilon_{\rm max}$, $\epsilon_{\rm min})$ for Gumbel-softmax (GS) & $(0.5,0.1)$ & $(0.75,0.1)$ \\
		\hline
		GS annealing parameter $r$ & $5\times 10^{-3}$  & $5\times 10^{-3}$ \\
		\hline
	\end{tabular}
	\label{CG_arch}
\end{table}

\subsection{Training convergence}
In Tab.~\ref{RSMIne_training_ising_dimer} we tabulate typical values for the parameters chosen for the training: specifically, we give the learning rate, sample sizes, mini-batch sizes, total number of iterations, and the total runtimes until convergence for the 2D Ising and 2D interacting dimer systems (which are extremely short). We have found that separating the full sample dataset into mini-batches to be used in a single iteration of training greatly enhances the performance. Also note that have used the same learning rate for both systems. 

\begin{table*}[]
	\centering
	\caption{\textbf{RSMI-NE training parameters for 2D Ising anti-ferromagnet and 2D interacting dimer model on a square lattice.}}
	\begin{tabular}{|c||c|c|c|}
		\hline
		model & 2D Ising ($T=T_c$)  & 2D dimer ($T<T_{\ssm BKT}$) & 2D dimer ($T>T_{\ssm BKT}$) \\
		\hline
		\hline
		sampling size& 20000  & 32400 & 32400  \\
		\hline
		mini-batch size& 100 & 400 & 400 \\
		\hline
		number of epochs & 30 & 15 & 50  \\
		\hline
		learning rate & $5\times10^{-3}$ & $5\times10^{-3}$ & $5\times10^{-3}$   \\
		\hline
		total runtime & \SI{35}{\s}  & \SI{8}{\s} & \SI{34}{\s}  \\
		\hline
	\end{tabular}
	\label{RSMIne_training_ising_dimer}
\end{table*}
\subsection{Criteria for convergence: halting the training}\label{convergence}
We typically repeat the training epochs described in Alg. \ref{RSMI_alg_training} until the estimated RSMI value converges. That is, if the RSMI estimate at the $t$'th epoch is $\tilde{I}^{(t)}_{\Lambda}(\mathcal{H}:\mathcal{E})$, we halt the training if
\begin{equation}
	\tilde{I}^{(t)}_{\Lambda}(\mathcal{H}:\mathcal{E}) - \tilde{I}^{(t-1)}_{\Lambda}(\mathcal{H}:\mathcal{E}) \leq \Delta,
\end{equation}
where $\Delta$ is a convergence threshold. 

While this simple halting condition is usually sufficient, there are cases where the convergence is slow because of the oscillatory behaviour of the RSMI estimate series. We can suppress the oscillatory behaviour by replacing the most recent estimate by its average with the previous estimate, that is:
\begin{equation}
	\tilde{I}^{(t+1)}_{\Lambda}(\mathcal{H}:\mathcal{E}) \leftarrow \alpha \tilde{I}^{(t+1)}_{\Lambda}(\mathcal{H}:\mathcal{E})  + (1-\alpha) \tilde{I}^{(t)}_{\Lambda}(\mathcal{H}:\mathcal{E}), \hspace{0.5 cm},
\end{equation}
with $\alpha \in \left [ 0,1\right[$ and use the halting condition above.

\paragraph{BFGS solver on the converged model.} A more sophisticated method to reach convergence is to apply quasi-Newton non-linear solvers to the trained RSMI \textit{ansatz} $\tilde{I}^{(t)}_{\Lambda}(\mathcal{H}:\mathcal{E})$. In this direction, we optionally use the Broyden-Fletcher-Goldfarb-Shanno (BFGS) solver to refine the most recent estimate:
\begin{equation}
	\tilde{I}^{(t+1)}_{\Lambda}(\mathcal{H}:\mathcal{E}) \leftarrow {\rm BFGS}[\tilde{I}^{(t+1)}_{\Lambda}(\mathcal{H}:\mathcal{E})].
\end{equation}
(Interested reader is directed to Refs.~\onlinecite{10.1093/imamat/6.3.222,BFGS} for the details of the BFGS algorithm.) The BFGS method uses a positive definite approximation to the Hessian matrix to find the search direction. Due to the large number of parameters of the optimisation problem, it is more convenient to apply the BFGS solver to refine only the coarse-graining parameters $\Lambda$, whilst keeping the parameters of the InfoNCE \textit{ansatz} $f(h,e)$ fixed.

\section{Generating dimer model samples with the directed loop Monte Carlo algorithm}\label{directedloopMC}
We have implemented the MC directed-loop algorithm \index{directed-loop algorithm} (DLA), which was first introduced by Sandvik and Moessner,\cite{PhysRevB.73.144504} to sample configurations from the interacting dimer model. Starting from a valid dimer configuration (at an arbitrary temperature), the DLA provides non-local moves in the configuration space of dimers on a lattice by changing positions of dimers along a closed loop (or a ``worm"). The high efficiency of this algorithm is due to the fact that the closed loops can be quite large, and they are formed according to a local detailed balance condition, without a further acceptance criterion.

One MC sweep of DLA comprises the following steps:
\begin{enumerate}
	\item Place a \textit{worm} on a valid dimer configuration at a random lattice site $i$. As we shall see, the worm is constituted by two monomer defects at its tail and its head.
	\item Move the head of the worm to site $j$, which is connected to site $i$ by a dimer in the background configuration, denoted by $(ij)$. Remove the dimer $(ij)$, thereby leaving the sites $i$ and $j$ with no dimers attached to them, \emph{i.e.}~as monomer defects.
	\item According to local detailed balance condition, randomly select one of the nearest neighbours of $j$, say site $k$. Then move the head of the worm to site $k$ and put a dimer $(jk)$ in between.
	\item Repeat $2-3$ until the worm becomes a closed loop, \textit{i.e.} $i=k$. Upon closing the loop, we get a valid dimer configuration as the monomer defects overlap and are annihilated.
\end{enumerate}
We provide the pseudocode describing a single sub-sweep of the DLA in Alg.~\ref{DLA_subsweep}:

\begin{algorithm*}[]
	\caption{Single sub-sweep of the directed loop algorithm}
	\label{DLA_subsweep}
	\begin{algorithmic}[1]
		\State $X \leftarrow$ PW algorithm \Comment{Generate random free dimer configuration}
		\While{unvisited sites in lattice $\neq \emptyset$}
		\State $i_0\leftarrow $ random integer $\in [1,L^2]$
		\While{$k\neq i_0$}
		\State $j \leftarrow$ site connected to $i$ by a dimer
		\State remove the dimer between sites $i$ and $j$
		\For{$nn$ in nearest neighbours of $j$}
		\State compute $w_{(nn\hspace{0.05cm}j)}$
		\EndFor
		\State $\mathbf{p}^j\leftarrow$ solve DLE using $w_{(nn\hspace{0.05cm}j)}$ and get scattering matrix \Comment{see Eq.~(\ref{scattering_matrix})}
		\State $k\leftarrow$ sample from $p^j_{i\to k}$
		\State add a dimer between sites $j$ and $k$
		\State $i\leftarrow k$
		\EndWhile
		\EndWhile
	\end{algorithmic}
\end{algorithm*}

The detailed balance conditions for the transition probabilities of the worm's head's position are determined by the fugacity of the local dimer configurations, with which they contribute to the partition function. The fugacity is given by the (unnormalised) weight of a dimer $(ij)$ defined as:
\begin{equation}\label{intdimer_weights}
	w_{(ij)}:=\exp(N_{(ij)}/T),
\end{equation}
where $N_{(ij)}\in \{0,1,2\}$ is the number of nearest neighbours of site $i$ which has a dimer parallel to $(ij)$.

The transition probabilities for moving the head of the worm from $i$ to $k$ [or replacing the dimer $(ij)$ with $(jk)$] must satisfy the detailed balance condition:
\begin{equation}\label{directed_loop_eqs}
	p[(ij)\to(jk)] w_{(ij)} = p[(jk)\to(ij)] w_{(jk)},
\end{equation}
which are also known as the directed-loop equations (DLEs).

The Eqs.~\ref{directed_loop_eqs} are under-determined: with the normalisation condition they constrain only 10 of the elements of the $4\times 4$ scattering matrix. It is common practice to specify the remaining transition weights by imposing minimisation of the so-called \textit{bounce} processes $p[(ij)\to(ji)] $. This allows to avoid trivial back-tracking moves of the worm, leading to the longest possible loops. As suggested by Alet \textit{et al.},\cite{PhysRevE.74.041124} linear programming (LP) techniques can be used to find the solution of the DLEs that minimise the bounce probabilities. In what follows we explain how the directed loop equations with minimal bounce probabilities can be formulated as an LP problem.

\subsection{Linear programming formulation of the directed loop equations}
We will formulate of the directed loop equations (DLEs) as a linear programming (LP) problem. Our task is to derive a Markov chain transition probability:
\begin{equation}
	p_{i\to k}^j := p[(ij)\to(jk)],
\end{equation}
to move the head of the worm from vertex $i$ to vertex $j$ using the DLEs. Let us define:
\begin{equation}
	a^j_{i\to k}:=p_{i\to k}^j w_{(ij)},
\end{equation}
with the unnormalised weights $w_{(ij)}$ defined as in Eq.~(\ref{intdimer_weights}).  Since the transition probabilities are normalised, it satisfies the relation:
\begin{equation}\label{a_w_relation}
	\sum_{k\in {\rm nn}(j)} a^j_{i\to k} = w_{(ij)}, \hspace{0.5cm} \forall i,
\end{equation}
where the sum is taken over the four nearest neighbours of site $j$, ${\rm nn}(j):=\{R(i),L(i),B(i),U(i)\}$.

The DLEs are simply given by the conditions of local detailed balance:
\begin{equation}
	a^j_{i\to k} = a^j_{k\to i},
\end{equation}
and the normalisation $\sum_{k\in {\rm nn}(i)} p_{i\to k}^j=1$ for all $i$. Crucially, the DLEs are under-determined and we can impose the restriction on the solution that the trivial bounce events are suppressed as much as possible. This implies minimising the objective:
\begin{align}
	f:&=\sum_{i\in{\rm nn}(j)} a_{i\to i}^j\nonumber\\
	&=C-\sum_{i \in{\rm nn}(j)} \sum_{k\neq i } a_{i\to k}^j,
\end{align}
under the normalisation constraint $w_{(ij)}\leq \sum_{k\neq i} a^j_{i\to k}$. Here, $C$ is a constant in $a^j_{i\to k}$, and the second line in the expression for $f$ follows from Eq.~(\ref{a_w_relation}). 

Since we work on a square lattice, we can define a $4\times 4$ scattering matrix:
\begin{equation}\label{scattering_matrix}
	\mathbf{p}^j:=\begin{pmatrix}
		p_{R(j)R(j)} & p_{R(j)L(j)} & p_{R(j)B(j)} & p_{R(j)U(j)} \\
		p_{R(j)L(j)} & p_{L(j)L(j)} & p_{L(j)B(j)} & p_{L(j)U(j)} \\
		p_{B(j)R(j)} & p_{B(j)L(j)} & p_{B(j)B(j)} & p_{B(j)U(j)} \\
		p_{U(j)R(j)} & p_{U(j)L(j)} & p_{U(j)B(j)} & p_{U(j)U(j)}
	\end{pmatrix},
\end{equation}
where 10 of the elements are restricted by the DLEs and the normalisation. As for the 6 free parameters, we select the following subset:
\begin{align}
	\mathbf{x}^{\rm T}:=(&a_{R(j)R(j)}, a_{R(j)B(j)}, a_{R(j)U(j)},\nonumber\\ 
	&a_{L(j)B(j)}, a_{L(j)U(j)}, a_{B(j)U(j)}).
\end{align}

LP solves problems of the form:
\begin{equation}\label{LP}
	\max_{\mathbf{x}} \mathbf{c}^{\rm T} \mathbf{x}, \hspace{0.5cm} \text{such that}  \hspace{0.5cm}\mathbf{A}\cdot \mathbf{x}\leq \mathbf{b}.
\end{equation}
Specifically, for the DLEs, we have:
\begin{equation}\label{LP_ing1}
	\min_{\mathbf{x}}f= \max_{\mathbf{x}}\left(\sum_{i, k\neq i} a^j_{i\to k} \right) = \max_{\mathbf{x}}  \underbrace{(1, 1, 1, 1, 1, 1)}_{=:\mathbf{c}^{\rm T}} \mathbf{x},
\end{equation}
and the normalisation constraint is imposed by using:
\begin{equation}\label{LP_ing2}
	\mathbf{A}=\begin{pmatrix}
		1&1&1&0&0&0\\
		1&0&0&1&1&0\\
		0&1&0&1&0&1\\
		0&0&1&0&1&1
	\end{pmatrix} \hspace{0.5cm} \text{and} \hspace{0.5cm} \mathbf{b}=\begin{pmatrix}
		w_{(R(j)j)}\\
		w_{(L(j)j)}\\
		w_{(B(j)j)}\\
		w_{(U(j)j)}
	\end{pmatrix}.
\end{equation}
Having identified the matrices in Eqs.~\ref{LP_ing1} and \ref{LP_ing2} we can proceed to solve the linear system of Eqs.~\ref{LP}, \textit{e.g.}~via the simplex method, using standard LP libraries. Note that the efficiency of the algorithm can be significantly increased by pre-computing all possible scattering matrices at the outset, and using the tabulated values during the construction of the loops.

Observe that in order to use the DLA we still need a valid dimer configuration to start with, which is non-trivial to generate for large systems. Even though finding a valid dimer covering on an arbitrary graph is difficult, for the square lattice it is possible to efficiently (time complexity linear with system size) get random free dimer configurations using the Propp-Wilson (PW) algorithm (also known as coupling-from-the-past method) by generating loop-erased random walks.\cite{kenyon1999trees} This method exploits a bijection between spanning trees of a undirected graph and valid dimer coverings. Our implementation for sampling from the interacting dimer model begins by generating a random dimer configuration using the PW algorithm.

A single MC sweep comprises constructing several closed loops (worms) until all sites on the lattice are visited at least once. To reduce the autocorrelations further, our single sweep consists of multiple sub-sweeps. The length of the worms get smaller at low temperatures, and the updates essentially become local single-dimer flips, at best. Hence, in order get uncorrelated samples, the number of sub-sweeps has to grow as the temperature is reduced. Moreover, to ensure the balance within all 4 columnar configurations at low temperatures, we performed 180 runs with different random PW initialisations, using 200 sweeps for each temperature. Our simulations were carried out on lattices with periodic boundary conditions. 

%% file: long_manuscript.bbl
\begin{thebibliography}{10}
\expandafter\ifx\csname url\endcsname\relax
  \def\url#1{\texttt{#1}}\fi
\expandafter\ifx\csname urlprefix\endcsname\relax\def\urlprefix{URL }\fi
\providecommand{\bibinfo}[2]{#2}
\providecommand{\eprint}[2][]{\url{#2}}

\bibitem{pt_multiscale}
\bibinfo{author}{{Schieber}, J.~D.} \& \bibinfo{author}{{H{\"u}tter}, M.}
\newblock \bibinfo{title}{{Multiscale modeling beyond equilibrium}}.
\newblock \emph{\bibinfo{journal}{Physics Today}}
  \textbf{\bibinfo{volume}{73}}, \bibinfo{pages}{36--42}
  (\bibinfo{year}{2020}).

\bibitem{PhysicsPhysiqueFizika.2.263}
\bibinfo{author}{Kadanoff, L.~P.}
\newblock \bibinfo{title}{{Scaling laws for Ising models near ${T}_{c}$}}.
\newblock \emph{\bibinfo{journal}{Phys. Phys. Fiz.}}
  \textbf{\bibinfo{volume}{2}}, \bibinfo{pages}{263--272}
  (\bibinfo{year}{1966}).

\bibitem{Wilson1974}
\bibinfo{author}{Wilson, K.~G.} \& \bibinfo{author}{Kogut, J.}
\newblock \bibinfo{title}{The renormalization group and the $\epsilon$
  expansion}.
\newblock \emph{\bibinfo{journal}{Phys. Rep.}} \textbf{\bibinfo{volume}{12}},
  \bibinfo{pages}{75 -- 199} (\bibinfo{year}{1974}).

\bibitem{Wilson1975}
\bibinfo{author}{Wilson, K.~G.}
\newblock \bibinfo{title}{{The renormalization group: Critical phenomena and
  the Kondo problem}}.
\newblock \emph{\bibinfo{journal}{Rev. Mod. Phys.}}
  \textbf{\bibinfo{volume}{47}}, \bibinfo{pages}{773--840}
  (\bibinfo{year}{1975}).

\bibitem{Fisher1998}
\bibinfo{author}{Fisher, M.~E.}
\newblock \bibinfo{title}{Renormalization group theory: Its basis and
  formulation in statistical physics}.
\newblock \emph{\bibinfo{journal}{Rev. Mod. Phys.}}
  \textbf{\bibinfo{volume}{70}}, \bibinfo{pages}{653--681}
  (\bibinfo{year}{1998}).

\bibitem{PhysRevB.4.3184}
\bibinfo{author}{Wilson, K.~G.}
\newblock \bibinfo{title}{Renormalization group and critical phenomena. {II}.
  phase-space cell analysis of critical behavior}.
\newblock \emph{\bibinfo{journal}{Phys. Rev. B}} \textbf{\bibinfo{volume}{4}},
  \bibinfo{pages}{3184--3205} (\bibinfo{year}{1971}).

\bibitem{PhysRevLett.30.1343}
\bibinfo{author}{Gross, D.~J.} \& \bibinfo{author}{Wilczek, F.}
\newblock \bibinfo{title}{Ultraviolet behavior of non-{Abelian} gauge
  theories}.
\newblock \emph{\bibinfo{journal}{Phys. Rev. Lett.}}
  \textbf{\bibinfo{volume}{30}}, \bibinfo{pages}{1343--1346}
  (\bibinfo{year}{1973}).

\bibitem{PhysRevLett.69.534}
\bibinfo{author}{Fisher, D.~S.}
\newblock \bibinfo{title}{Random transverse field {Ising} spin chains}.
\newblock \emph{\bibinfo{journal}{Phys. Rev. Lett.}}
  \textbf{\bibinfo{volume}{69}}, \bibinfo{pages}{534--537}
  (\bibinfo{year}{1992}).

\bibitem{PhysRevB.51.6411}
\bibinfo{author}{Fisher, D.~S.}
\newblock \bibinfo{title}{Critical behavior of random transverse-field {Ising}
  spin chains}.
\newblock \emph{\bibinfo{journal}{Phys. Rev. B}} \textbf{\bibinfo{volume}{51}},
  \bibinfo{pages}{6411--6461} (\bibinfo{year}{1995}).

\bibitem{PhysRevB.50.3799}
\bibinfo{author}{Fisher, D.~S.}
\newblock \bibinfo{title}{Random antiferromagnetic quantum spin chains}.
\newblock \emph{\bibinfo{journal}{Phys. Rev. B}} \textbf{\bibinfo{volume}{50}},
  \bibinfo{pages}{3799--3821} (\bibinfo{year}{1994}).

\bibitem{IGLOI2005277}
\bibinfo{author}{Igl{\'o}i, F.} \& \bibinfo{author}{Monthus, C.}
\newblock \bibinfo{title}{Strong disorder {RG} approach of random systems}.
\newblock \emph{\bibinfo{journal}{Phys. Rep.}} \textbf{\bibinfo{volume}{412}},
  \bibinfo{pages}{277--431} (\bibinfo{year}{2005}).

\bibitem{Koch-Janusz2018}
\bibinfo{author}{Koch-Janusz, M.} \& \bibinfo{author}{Ringel, Z.}
\newblock \bibinfo{title}{{Mutual information, neural networks and the
  renormalization group}}.
\newblock \emph{\bibinfo{journal}{Nat. Phys.}} \textbf{\bibinfo{volume}{14}},
  \bibinfo{pages}{578--582} (\bibinfo{year}{2018}).

\bibitem{optimalRSMI}
\bibinfo{author}{Lenggenhager, P.~M.}, \bibinfo{author}{G\"okmen, D.~E.},
  \bibinfo{author}{Ringel, Z.}, \bibinfo{author}{Huber, S.~D.} \&
  \bibinfo{author}{Koch-Janusz, M.}
\newblock \bibinfo{title}{Optimal renormalization group transformation from
  information theory}.
\newblock \emph{\bibinfo{journal}{Phys. Rev. X}} \textbf{\bibinfo{volume}{10}},
  \bibinfo{pages}{011037} (\bibinfo{year}{2020}).

\bibitem{Gordon2020}
\bibinfo{author}{Gordon, A.}, \bibinfo{author}{Banerjee, A.},
  \bibinfo{author}{Koch-Janusz, M.} \& \bibinfo{author}{Ringel, Z.}
\newblock \bibinfo{title}{Relevance in the renormalization group and in
  information theory}.
\newblock \emph{\bibinfo{journal}{Phys. Rev. Lett.}}
  \textbf{\bibinfo{volume}{126}}, \bibinfo{pages}{240601}
  (\bibinfo{year}{2021}).

\bibitem{poole2019variational}
\bibinfo{author}{Poole, B.}, \bibinfo{author}{Ozair, S.},
  \bibinfo{author}{van~den Oord, A.}, \bibinfo{author}{Alemi, A.~A.} \&
  \bibinfo{author}{Tucker, G.}
\newblock \bibinfo{title}{On variational bounds of mutual information}
  (\bibinfo{year}{2019}).
\newblock \eprint{arXiv:1905.06922}.

\bibitem{belghazi2018mine}
\bibinfo{author}{Belghazi, M.~I.} \emph{et~al.}
\newblock \bibinfo{title}{{MINE: Mutual Information Neural Estimation}}
  (\bibinfo{year}{2018}).
\newblock \eprint{arXiv:1801.04062}.

\bibitem{rsmine_letter}
\bibinfo{author}{G{\"o}kmen, D.~E.}, , \bibinfo{author}{Ringel, Z.},
  \bibinfo{author}{Huber, S.~D.} \& \bibinfo{author}{Koch-Janusz, M.}
\newblock \bibinfo{title}{Statistical physics through the lens of real-space
  mutual information} (\bibinfo{year}{2021}).
\newblock \eprint{arXiv:2101.11633}.

\bibitem{vanEnter1993}
\bibinfo{author}{van Enter, A. C.~D.}, \bibinfo{author}{Fern{\'a}ndez, R.} \&
  \bibinfo{author}{Sokal, A.~D.}
\newblock \bibinfo{title}{{Regularity properties and pathologies of
  position-space renormalization-group transformations: Scope and limitations
  of Gibbsian theory}}.
\newblock \emph{\bibinfo{journal}{J. Stat. Phys.}}
  \textbf{\bibinfo{volume}{72}}, \bibinfo{pages}{879--1167}
  (\bibinfo{year}{1993}).

\bibitem{PhysRevE.69.066138}
\bibinfo{author}{Kraskov, A.}, \bibinfo{author}{St\"ogbauer, H.} \&
  \bibinfo{author}{Grassberger, P.}
\newblock \bibinfo{title}{Estimating mutual information}.
\newblock \emph{\bibinfo{journal}{Phys. Rev. E}} \textbf{\bibinfo{volume}{69}},
  \bibinfo{pages}{066138} (\bibinfo{year}{2004}).

\bibitem{PhysRevE.52.6841}
\bibinfo{author}{Wolpert, D.~H.} \& \bibinfo{author}{Wolf, D.~R.}
\newblock \bibinfo{title}{Estimating functions of probability distributions
  from a finite set of samples}.
\newblock \emph{\bibinfo{journal}{Phys. Rev. E}} \textbf{\bibinfo{volume}{52}},
  \bibinfo{pages}{6841--6854} (\bibinfo{year}{1995}).

\bibitem{doi:10.1002/cpa.3160360204}
\bibinfo{author}{Donsker, M.~D.} \& \bibinfo{author}{Varadhan, S. R.~S.}
\newblock \bibinfo{title}{{Asymptotic evaluation of certain {Markov} process
  expectations for large time. {IV}}}.
\newblock \emph{\bibinfo{journal}{Commun. Pure Appl. Math.}}
  \textbf{\bibinfo{volume}{36}}, \bibinfo{pages}{183--212}
  (\bibinfo{year}{1983}).

\bibitem{NIPS2003_2410}
\bibinfo{author}{Barber, D.} \& \bibinfo{author}{Agakov, F.~V.}
\newblock \bibinfo{title}{Information maximization in noisy channels: A
  variational approach}.
\newblock In \bibinfo{editor}{Thrun, S.}, \bibinfo{editor}{Saul, L.~K.} \&
  \bibinfo{editor}{Sch\"{o}lkopf, B.} (eds.) \emph{\bibinfo{booktitle}{Advances
  in Neural Information Processing Systems 16}}, \bibinfo{pages}{201--208}
  (\bibinfo{publisher}{MIT Press}, \bibinfo{year}{2004}).

\bibitem{5605355}
\bibinfo{author}{Nguyen, X.}, \bibinfo{author}{Wainwright, M.~J.} \&
  \bibinfo{author}{Jordan, M.~I.}
\newblock \bibinfo{title}{Estimating divergence functionals and the likelihood
  ratio by penalized convex risk minimization}.
\newblock In \bibinfo{editor}{Platt, J.~C.}, \bibinfo{editor}{Koller, D.},
  \bibinfo{editor}{Singer, Y.} \& \bibinfo{editor}{Roweis, S.~T.} (eds.)
  \emph{\bibinfo{booktitle}{Advances in Neural Information Processing Systems
  20}}, \bibinfo{pages}{1089--1096} (\bibinfo{publisher}{Neural Information
  Processing Systems Foundation}, \bibinfo{year}{2009}).

\bibitem{nowozin2016fgan}
\bibinfo{author}{Nowozin, S.}, \bibinfo{author}{Cseke, B.} \&
  \bibinfo{author}{Tomioka, R.}
\newblock \bibinfo{title}{f-{GAN}: Training generative neural samplers using
  variational divergence minimization} (\bibinfo{year}{2016}).
\newblock \eprint{arXiv:1606.00709}.

\bibitem{Nir30234}
\bibinfo{author}{Nir, A.}, \bibinfo{author}{Sela, E.}, \bibinfo{author}{Beck,
  R.} \& \bibinfo{author}{Bar-Sinai, Y.}
\newblock \bibinfo{title}{Machine-learning iterative calculation of entropy for
  physical systems}.
\newblock \emph{\bibinfo{journal}{Proc. Natl. Acad. Sci. U.S.A.}}
  \textbf{\bibinfo{volume}{117}}, \bibinfo{pages}{30234--30240}
  (\bibinfo{year}{2020}).

\bibitem{oord2018representation}
\bibinfo{author}{van~den Oord, A.}, \bibinfo{author}{Li, Y.} \&
  \bibinfo{author}{Vinyals, O.}
\newblock \bibinfo{title}{Representation learning with contrastive predictive
  coding} (\bibinfo{year}{2018}).
\newblock \eprint{arXiv:1807.03748}.

\bibitem{899a65b4919f47c8a06d115df85dad11}
\bibinfo{author}{Gutmann, M.} \& \bibinfo{author}{Hyv{\"a}rinen, A.}
\newblock \bibinfo{title}{Noise-contrastive estimation: A new estimation
  principle for unnormalized statistical models}.
\newblock In \bibinfo{editor}{Teh, Y.~W.} \& \bibinfo{editor}{Titterington, M.}
  (eds.) \emph{\bibinfo{booktitle}{Proceedings of the Thirteenth International
  Conference on Artificial Intelligence and Statistics}},
  vol.~\bibinfo{volume}{9} of \emph{\bibinfo{series}{Proceedings of Machine
  Learning Research}}, \bibinfo{pages}{297--304} (\bibinfo{publisher}{JMLR
  Workshop and Conference Proceedings}, \bibinfo{year}{2010}).

\bibitem{Goodfellow-et-al-2016}
\bibinfo{author}{Goodfellow, I.}, \bibinfo{author}{Bengio, Y.} \&
  \bibinfo{author}{Courville, A.}
\newblock \emph{\bibinfo{title}{Deep Learning}} (\bibinfo{publisher}{MIT
  Press}, \bibinfo{year}{2016}).
\newblock \urlprefix\url{http://www.deeplearningbook.org}.

\bibitem{PhysRevX.10.031056}
\bibinfo{author}{Lu, P.~Y.}, \bibinfo{author}{Kim, S.} \&
  \bibinfo{author}{Solja\ifmmode \check{c}\else
  \v{c}\fi{}i\ifmmode~\acute{c}\else \'{c}\fi{}, M.}
\newblock \bibinfo{title}{Extracting interpretable physical parameters from
  spatiotemporal systems using unsupervised learning}.
\newblock \emph{\bibinfo{journal}{Phys. Rev. X}} \textbf{\bibinfo{volume}{10}},
  \bibinfo{pages}{031056} (\bibinfo{year}{2020}).

\bibitem{jang2016categorical}
\bibinfo{author}{Jang, E.}, \bibinfo{author}{Gu, S.} \& \bibinfo{author}{Poole,
  B.}
\newblock \bibinfo{title}{{Categorical Reparameterization with Gumbel-Softmax}}
  (\bibinfo{year}{2016}).
\newblock \eprint{arXiv:1611.01144}.

\bibitem{NIPS2014_5449}
\bibinfo{author}{Maddison, C.~J.}, \bibinfo{author}{Tarlow, D.} \&
  \bibinfo{author}{Minka, T.}
\newblock \bibinfo{title}{$\text{A}^\ast$ sampling}.
\newblock In \bibinfo{editor}{Ghahramani, Z.}, \bibinfo{editor}{Welling, M.},
  \bibinfo{editor}{Cortes, C.}, \bibinfo{editor}{Lawrence, N.~D.} \&
  \bibinfo{editor}{Weinberger, K.~Q.} (eds.) \emph{\bibinfo{booktitle}{Advances
  in Neural Information Processing Systems 27}}, \bibinfo{pages}{3086--3094}
  (\bibinfo{publisher}{Curran Associates, Inc.}, \bibinfo{year}{2014}).

\bibitem{gumbel1954}
\bibinfo{author}{Gumbel, E.~J.}
\newblock \bibinfo{title}{The maxima of the mean largest value and of the
  range}.
\newblock \emph{\bibinfo{journal}{Ann. Math. Stat.}}
  \textbf{\bibinfo{volume}{25}}, \bibinfo{pages}{76--84}
  (\bibinfo{year}{1954}).

\bibitem{kingma2014adam}
\bibinfo{author}{Kingma, D.~P.} \& \bibinfo{author}{Ba, J.}
\newblock \bibinfo{title}{Adam: A method for stochastic optimization}
  (\bibinfo{year}{2014}).
\newblock \eprint{arXiv:1412.6980}.

\bibitem{RevModPhys.91.045002}
\bibinfo{author}{Carleo, G.} \emph{et~al.}
\newblock \bibinfo{title}{Machine learning and the physical sciences}.
\newblock \emph{\bibinfo{journal}{Rev. Mod. Phys.}}
  \textbf{\bibinfo{volume}{91}}, \bibinfo{pages}{045002}
  (\bibinfo{year}{2019}).

\bibitem{doi:10.1080/23746149.2020.1797528}
\bibinfo{author}{Carrasquilla, J.}
\newblock \bibinfo{title}{Machine learning for quantum matter}.
\newblock \emph{\bibinfo{journal}{Advances in Physics: X}}
  \textbf{\bibinfo{volume}{5}}, \bibinfo{pages}{1797528}
  (\bibinfo{year}{2020}).
\newblock \eprint{https://doi.org/10.1080/23746149.2020.1797528}.

\bibitem{PhysRev.65.117}
\bibinfo{author}{Onsager, L.}
\newblock \bibinfo{title}{Crystal statistics. {I}. {A} two-dimensional model
  with an order-disorder transition}.
\newblock \emph{\bibinfo{journal}{Phys. Rev.}} \textbf{\bibinfo{volume}{65}},
  \bibinfo{pages}{117--149} (\bibinfo{year}{1944}).

\bibitem{Wilms_2011}
\bibinfo{author}{Wilms, J.}, \bibinfo{author}{Troyer, M.} \&
  \bibinfo{author}{Verstraete, F.}
\newblock \bibinfo{title}{Mutual information in classical spin models}.
\newblock \emph{\bibinfo{journal}{J. Stat. Mech.: Theory Exp.}}
  \textbf{\bibinfo{volume}{2011}}, \bibinfo{pages}{P10011}
  (\bibinfo{year}{2011}).

\bibitem{PhysRevE.87.022128}
\bibinfo{author}{Lau, H.~W.} \& \bibinfo{author}{Grassberger, P.}
\newblock \bibinfo{title}{{Information theoretic aspects of the two-dimensional
  Ising model}}.
\newblock \emph{\bibinfo{journal}{Phys. Rev. E}} \textbf{\bibinfo{volume}{87}},
  \bibinfo{pages}{022128} (\bibinfo{year}{2013}).

\bibitem{PhysRevE.74.041124}
\bibinfo{author}{Alet, F.}, \bibinfo{author}{Ikhlef, Y.},
  \bibinfo{author}{Jacobsen, J.~L.}, \bibinfo{author}{Misguich, G.} \&
  \bibinfo{author}{Pasquier, V.}
\newblock \bibinfo{title}{Classical dimers with aligning interactions on the
  square lattice}.
\newblock \emph{\bibinfo{journal}{Phys. Rev. E}} \textbf{\bibinfo{volume}{74}},
  \bibinfo{pages}{041124} (\bibinfo{year}{2006}).

\bibitem{PhysRevLett.94.235702}
\bibinfo{author}{Alet, F.} \emph{et~al.}
\newblock \bibinfo{title}{Interacting classical dimers on the square lattice}.
\newblock \emph{\bibinfo{journal}{Phys. Rev. Lett.}}
  \textbf{\bibinfo{volume}{94}}, \bibinfo{pages}{235702}
  (\bibinfo{year}{2005}).

\bibitem{fradkin_2013}
\bibinfo{author}{Fradkin, E.}
\newblock \emph{\bibinfo{title}{Field Theories of Condensed Matter Physics}}
  (\bibinfo{publisher}{Cambridge University Press}, \bibinfo{year}{2013}),
  \bibinfo{edition}{2} edn.

\bibitem{PhysRevLett.100.070502}
\bibinfo{author}{Wolf, M.~M.}, \bibinfo{author}{Verstraete, F.},
  \bibinfo{author}{Hastings, M.~B.} \& \bibinfo{author}{Cirac, J.~I.}
\newblock \bibinfo{title}{Area laws in quantum systems: Mutual information and
  correlations}.
\newblock \emph{\bibinfo{journal}{Phys. Rev. Lett.}}
  \textbf{\bibinfo{volume}{100}}, \bibinfo{pages}{070502}
  (\bibinfo{year}{2008}).

\bibitem{bondesan2019learning}
\bibinfo{author}{Bondesan, R.} \& \bibinfo{author}{Lamacraft, A.}
\newblock \bibinfo{title}{Learning symmetries of classical integrable systems}
  (\bibinfo{year}{2019}).
\newblock \eprint{arXiv:1906.04645}.

\bibitem{PhysRevE.63.036114}
\bibinfo{author}{Rajesh, R.} \& \bibinfo{author}{Majumdar, S.~N.}
\newblock \bibinfo{title}{Exact phase diagram of a model with aggregation and
  chipping}.
\newblock \emph{\bibinfo{journal}{Phys. Rev. E}} \textbf{\bibinfo{volume}{63}},
  \bibinfo{pages}{036114} (\bibinfo{year}{2001}).

\bibitem{PhysRevE.79.041925}
\bibinfo{author}{Creutzig, F.}, \bibinfo{author}{Globerson, A.} \&
  \bibinfo{author}{Tishby, N.}
\newblock \bibinfo{title}{Past-future information bottleneck in dynamical
  systems}.
\newblock \emph{\bibinfo{journal}{Phys. Rev. E}} \textbf{\bibinfo{volume}{79}},
  \bibinfo{pages}{041925} (\bibinfo{year}{2009}).

\bibitem{Peters2006}
\bibinfo{author}{Peters, O.} \& \bibinfo{author}{Neelin, J.~D.}
\newblock \bibinfo{title}{{Critical phenomena in atmospheric precipitation}}.
\newblock \emph{\bibinfo{journal}{Nat. Phys.}} \textbf{\bibinfo{volume}{2}},
  \bibinfo{pages}{393--396. 5 p} (\bibinfo{year}{2006}).

\bibitem{10.1093/imamat/6.3.222}
\bibinfo{author}{Broyden, C.~G.}
\newblock \bibinfo{title}{{The Convergence of a Class of Double-rank
  Minimization Algorithms: 2. The New Algorithm}}.
\newblock \emph{\bibinfo{journal}{IMA J. Appl. Math.}}
  \textbf{\bibinfo{volume}{6}}, \bibinfo{pages}{222--231}
  (\bibinfo{year}{1970}).

\bibitem{BFGS}
\bibinfo{author}{Nocedal, J.} \& \bibinfo{author}{Wright, S.}
\newblock \emph{\bibinfo{title}{Numerical Optimization}}
  (\bibinfo{publisher}{Springer Series in Operations Research},
  \bibinfo{year}{2006}).

\bibitem{PhysRevB.73.144504}
\bibinfo{author}{Sandvik, A.~W.} \& \bibinfo{author}{Moessner, R.}
\newblock \bibinfo{title}{Correlations and confinement in nonplanar
  two-dimensional dimer models}.
\newblock \emph{\bibinfo{journal}{Phys. Rev. B}} \textbf{\bibinfo{volume}{73}},
  \bibinfo{pages}{144504} (\bibinfo{year}{2006}).

\bibitem{kenyon1999trees}
\bibinfo{author}{Kenyon, R.~W.}, \bibinfo{author}{Propp, J.~G.} \&
  \bibinfo{author}{Wilson, D.~B.}
\newblock \bibinfo{title}{Trees and matchings}.
\newblock \emph{\bibinfo{journal}{Electron. J. Comb.}}
  \textbf{\bibinfo{volume}{7}}, \bibinfo{pages}{25--34} (\bibinfo{year}{2000}).

\end{thebibliography}
